\documentclass[letterpaper,11pt]{article}

\usepackage{amsmath}
\usepackage{amsfonts}
\usepackage{amssymb}
\usepackage{amsthm}
\usepackage[usenames]{color}
\usepackage{fullpage}
\usepackage{xspace}
\usepackage{url,ifthen}
\usepackage{srcltx}
\usepackage{multirow}
\usepackage{boxedminipage}
\usepackage[margin=1.2in]{geometry}
\usepackage{nicefrac}
\usepackage{xspace}
\usepackage[pdftex]{graphicx}
\usepackage[usenames]{color}
\usepackage{srcltx}
\definecolor{DarkGreen}{rgb}{0.1,0.5,0.1}
\definecolor{DarkRed}{rgb}{0.5,0.1,0.1}
\definecolor{DarkBlue}{rgb}{0.1,0.1,0.5}
\usepackage[small]{caption}

\usepackage[pdftex]{hyperref}
\hypersetup{
    unicode=false,          
    pdftoolbar=true,        
    pdfmenubar=true,        
    pdffitwindow=false,      
    pdfnewwindow=true,      
    colorlinks=true,       
    linkcolor=DarkBlue,          
    citecolor=DarkGreen,        
    filecolor=DarkRed,      
    urlcolor=DarkBlue,          
}

\def\draft{1}

\def\submit{0}

\ifnum\draft=1 
    \def\ShowAuthNotes{1}
\else
    \def\ShowAuthNotes{0}
\fi

\ifnum\submit=1
\newcommand{\forsubmit}[1]{#1}
\newcommand{\forreals}[1]{}
\else
\newcommand{\forreals}[1]{#1}
\newcommand{\forsubmit}[1]{}
\fi

\ifnum\ShowAuthNotes=1
\newcommand{\authnote}[2]{{ \footnotesize \bf{{\color{DarkRed}[#1's Note:}
{\color{DarkBlue}#2}]}}}
\else
\newcommand{\authnote}[2]{}
\fi

\newtheorem{theorem}{Theorem}[section]
\newtheorem{remark}[theorem]{Remark}
\newtheorem{lemma}[theorem]{Lemma}
\newtheorem{corollary}[theorem]{Corollary}

\newtheorem{claim}[theorem]{Claim}
\newtheorem{fact}[theorem]{Fact}

\theoremstyle{definition}
\newtheorem{definition}[theorem]{Definition}

\newcommand{\chapterref}[1]{\hyperref[ch:#1]{Chapter~\ref{ch:#1}}}

\newcommand{\claimref}[1]{\hyperref[claim:#1]{Claim~\ref{claim:#1}}}
\newcommand{\corollarylabel}[1]{\label{cor:#1}}
\newcommand{\corollaryref}[1]{\hyperref[cor:#1]{Corollary~\ref{cor:#1}}}
\newcommand{\definitionlabel}[1]{\label{def:#1}}
\newcommand{\definitionref}[1]{\hyperref[def:#1]{Definition~\ref{def:#1}}}
\newcommand{\equationlabel}[1]{\label{eq:#1}}
\newcommand{\equationref}[1]{\hyperref[eq:#1]{Equation~\ref{eq:#1}}}
\newcommand{\factlabel}[1]{\label{fact:#1}}
\newcommand{\factref}[1]{\hyperref[fact:#1]{Fact~\ref{fact:#1}}}
\newcommand{\figurelabel}[1]{\label{fig:#1}}
\newcommand{\figureref}[1]{\hyperref[fig:#1]{Figure~\ref{fig:#1}}}
\newcommand{\itemlabel}[1]{\label{item:#1}}
\newcommand{\itemref}[1]{\hyperref[item:#1]{Item~(\ref{item:#1})}}
\newcommand{\lemmalabel}[1]{\label{lem:#1}}
\newcommand{\lemmaref}[1]{\hyperref[lem:#1]{Lemma~\ref{lem:#1}}}

\newcommand{\propref}[1]{\hyperref[prop:#1]{Proposition~\ref{prop:#1}}}

\newcommand{\propositionref}[1]{\hyperref[prop:#1]{Proposition~\ref{prop:#1}}}

\newcommand{\remarkref}[1]{\hyperref[rem:#1]{Remark~\ref{rem:#1}}}
\newcommand{\sectionlabel}[1]{\label{sec:#1}}
\newcommand{\sectionref}[1]{\hyperref[sec:#1]{Section~\ref{sec:#1}}}
\newcommand{\theoremlabel}[1]{\label{thm:#1}}
\newcommand{\theoremref}[1]{\hyperref[thm:#1]{Theorem~\ref{thm:#1}}}

\usepackage[T1]{fontenc}
\usepackage{kpfonts}
\usepackage{microtype}

\newcommand{\Esymb}{\mathbb{E}}
\newcommand{\Psymb}{\mathbb{P}}
\newcommand{\Vsymb}{\mathbb{V}}
\DeclareMathOperator*{\E}{\Esymb}
\DeclareMathOperator*{\Var}{\Vsymb}
\DeclareMathOperator*{\ProbOp}{\Psymb}
\renewcommand{\Pr}{\ProbOp}

\newcommand{\mper}{\,.}
\newcommand{\mcom}{\,,}

\newcommand{\cA}{{\cal A}}

\newcommand{\cG}{{\cal G}}

\newcommand{\cM}{{\cal M}}

\newcommand{\defeq}{\stackrel{\small \mathrm{def}}{=}}
\renewcommand{\leq}{\leqslant}
\renewcommand{\le}{\leqslant}
\renewcommand{\geq}{\geqslant}
\renewcommand{\ge}{\geqslant}

\newcommand{\set}[1]{\{#1\}}

\newcommand{\Set}[1]{\left\{#1\right\}}

\newcommand\rd{\,\mathrm{d}}
\newcommand{\bits}{\{0,1\}}

\newcommand{\R}{\mathbb{R}}

\usepackage{bm}

\newcommand{\poly}{{\rm poly}}

\renewcommand{\epsilon}{\varepsilon}

\newcommand{\eps}{\epsilon}

\newcommand{\remove}[1]{}

\newcommand{\ignore}[1]{}

\newcommand{\tv}{\mathrm{tv}}

\title{How Robust are Linear Sketches to Adaptive Inputs?}
\author{Moritz Hardt\thanks{IBM Almaden Research. Email: {\tt mhardt@us.ibm.com}}\and
David P. Woodruff\thanks{IBM Almaden Research. Email: {\tt dpwoodru@us.ibm.com}}}

\begin{document}
\maketitle
\begin{abstract}
Linear sketches are powerful algorithmic tools that turn an
$n$-dimensional input into a concise lower-dimensional representation via a linear
transformation. Such sketches have seen a wide range of applications including
norm estimation over data streams, compressed sensing, and distributed
computing. In almost any realistic setting, however, a linear sketch faces the
possibility that its inputs are correlated with previous evaluations of the
sketch. Known techniques no longer guarantee the correctness of the output in
the presence of such correlations.  We therefore ask: Are linear sketches inherently non-robust
to adaptively chosen inputs?  We give a strong affirmative answer to this
question. Specifically, we show that \emph{no} linear sketch approximates the
Euclidean norm of its input to within an arbitrary multiplicative 
approximation factor on a
polynomial number of adaptively chosen inputs. The result remains true 
even if the dimension of the sketch is $d=n-o(n)$ and the sketch is given
unbounded computation time. Our result is based on an algorithm with running
time polynomial in $d$ that adaptively finds a distribution over inputs on
which the sketch is incorrect with constant probability.
Our result implies several corollaries for related problems including
$\ell_p$-norm estimation and compressed sensing. Notably, we resolve an open
problem in compressed sensing regarding the feasibility of
$\ell_2/\ell_2$-recovery guarantees in presence of computationally bounded 
adversaries.
\end{abstract}
\vfill
\thispagestyle{empty}
\pagebreak
\tableofcontents
\vfill
\pagebreak
%
%
%
%
\newcommand{\GammaD}{\Gamma}

\section{Introduction}
Recent years have witnessed an explosion in the amount of available data, such as that
in data warehouses, the internet, sensor networks, and transaction logs. The need to
process this data efficiently has led to the emergence of new fields, including compressed sensing,
data stream algorithms and distributed functional monitoring. A unifying technique in 
these fields is the use of linear sketches. This technique involves specifying a
distribution $\pi$ over linear maps $A:\mathbb{R}^n \rightarrow \mathbb{R}^r$ for a value $r \ll n$. A 
matrix $A$ is sampled from $\pi$. Then, in the online phase, a vector $x \in \mathbb{R}^n$
is presented to the algorithm, which maintains the ``sketch'' $Ax$. This
provides a concise summary of $x$, from which various queries about $x$ can be
approximately answered. The storage and number of linear measurements (rows of
$A$) required is  proportional to $r$. The goal is to minimize $r$ to
well-approximate a large class of queries with high probability. 

\paragraph{Applications of Linear Sketches.}
In compressed sensing the goal is to design a distribution $\pi$ so that 
for $A \sim \pi$, given a vector $x \in \mathbb{R}^n$, from $Ax$ one can output a vector $x'$ 
for which $\|x-x'\|_p \leq C\|x_{tail(k)}\|_q$, where $x_{tail(k)}$ denotes $x$ with its top $k$ coefficients 
(in magnitude) replaced with zero, $p$ and $q$ are norms, and $C > 1$ is an approximation
parameter. The scheme is considered efficient if $r \leq k \cdot \poly(\log n)$. 
There are two common models, the ``for all'' and ``for each'' models. In
the ``for all'' model, a single $A$ is chosen and is required, with high probability, to work
simultaneously for all $x \in \mathbb{R}^n$. In the ``for each'' model, the chosen $A$ is just required
to work with high probability for any fixed $x \in \mathbb{R}^n$. 

A related model is the turnstile model for data streams. Here an underlying vector
$x \in \mathbb{R}^n$ is initialized to $0^n$ and undergoes a long sequence of additive updates to its coordinates 
of the form $x_i \leftarrow x_i + \delta$. The algorithm is presented the updates one by one, and maintains
a summary of what it has seen. If the summary is a linear sketch $Ax$, then given an additive update to the $i$-th coordinate,
the summary can be updated by adding $\delta A_i$ to $Ax$, where $A_i$ is the $i$-th column of $A$. The best known 
algorithms for any problem in this model maintain a linear sketch. 
Starting with the work of Alon, Matias, and Szegedy \cite{AMS99}, problems 
such as approximating the $p$-norm $\|x\|_p = (\sum_{i=1}^n |x_i|^p)^{1/p}$
for $1 \leq p \leq \infty$ (also known as the frequency moments), 
the heavy hitters or largest
coordinates in $x$,
and many others have been considered; we refer the reader to \cite{IndykCourse,Muthu}. Often it is required that
the algorithm be able to query the sketch to approximate the statistic 
at intermediate points in the stream, rather than solely at the end of the stream. 

Other examples include distributed computing \cite{mfhh02} and functional monitoring \cite{cmy11}. Here there 
are $k$ parties $P^1, \ldots, P^k$, e.g., database servers or sensor networks, 
each with a local stream $\mathcal{S}^i$ of additive updates to a vector $x^i$. The goal
is to approximate statistics, such as those mentioned above, 
on the aggregate vector $x = \sum_{i=1}^k x^i$. If the parties share public randomness, they can agree upon a sketching matrix
$A$. Then, each party can locally compute $Ax^i$, from which $Ax$ can be computed using the linearity
of the sketch, namely, $Ax = A(x^1 + \cdots + x^k)$. The important measure is the communication complexity, 
which, since it suffices to exchange the sketches $Ax^i$, is proportional to $r$ rather than to $n$. 

\paragraph{Adaptively Chosen Inputs.}
One weakness with the models above is that they assume the sketching matrix $A$ is 
{\it independent} of the input vector $x$. As pointed out in recent papers \cite{ghrsw12,ghsww12}, there are
applications for which this is inadequate. Indeed, this occurs in situations for which the result of
performing a query on $Ax$ influences future updates to the vector $x$. One example given in \cite{ghrsw12} is
that of a grocery store, in which $x$ consists of transactions, and one uses $Ax$ to approximate the best
selling items. The store may update its inventory based on $Ax$, which in turn influences
future sales. A more adversarial example given in \cite{ghrsw12,ghsww12} is that of using a compressed sensing radar
on a ship to avoid a missile from an attacker. Based on $Ax$, the ship takes evasive
action, which is viewed by the attacker, and may change the attack. The matrix $A$ used by the radar 
cannot be changed between successive attacks for efficiency reasons. 
Another example arises in high frequency stock trading. Imagine Alice monitors a stream of orders on the stock market 
and places her own orders depending on statistics based on sketches. A competitor Charlie might have a commercial 
interest in leading Alice's algorithm astray by observing her orders and manipulating the input stream accordingly.
The question of sketching in adversarial environments was also introduced and motivated in the beautiful work of 
Mironov, Naor and Segev~\cite{MironovNS08} who provide several examples arising in multiparty sketching applications.
Even from a less adversarial point of view, it seems hard to argue that in
realistic settings there will be no correlation between the inputs to a linear
sketch and previous evaluations of it. Resilience to such correlations would
be a desirable robustness guarantee of a sketching algorithm.

A deterministic sketching matrix, e.g., in compressed sensing one that
satisfies the ``for all'' property above, would suffice to handle this kind of
feedback. Unfortunately, such sketches provably have much weaker error
guarantees.
Indeed, if one wants the number $r$ of measurements to be on the order of $k \cdot \poly(\log n)$, then the best 
one can hope for is that for all $x \in \mathbb{R}^n$, from $Ax$ one can output $x'$ for which
$\|x-x'\|_2 \leq \frac{\eps}{\sqrt{k}} \|x_{tail(k)}\|_1$, which is known as the $\ell_2/\ell_1$ error guarantee. 
However, if one allows the ``for each'' property,
then there are distributions $\pi$ over sketching matrices $A$ for which for any fixed $x \in \mathbb{R}^n$,
from $Ax$ one can output $x'$ for which $\|x-x'\|_2 \leq (1+\eps)\|x_{tail(k)}\|_2$ with high probability
(over $A \sim \pi$), which is known as the
$\ell_2/\ell_2$ error guarantee. One can verify that the
second guarantee is much stronger than the first; indeed, for constant $\eps$ and $k = 1$, if 
$x = (\sqrt{n}, \pm 1, \pm 1, \ldots, \pm 1)$, then with the $\ell_2/\ell_1$ guarantee, an output of $x' = 0^n$ is valid,
while for the $\ell_2/\ell_2$ guarantee, $x'_1$ must either be large or many coordinates of $x'$ must agree in sign
with those of $x$. 

An important open question, indeed, the first open question\footnote{See 
\url{http://ls2-www.cs.tu-dortmund.de/streamingWS2012/slides/open.problems_dortmund2012.pdf}} in the 
``Open Questions from the Workshop on Algorithms for Data Streams 2012 at Dortmund'', is whether or not it is 
possible to achieve the $\ell_2/\ell_2$ guarantee for probabilistic polynomial time adversaries with limited
information about $A$. The weakest possible information an adversary can have about $A$ is through black box
queries. Formally, given a sketch $Ax$, there is a function $f(Ax)$ for which its output satisfies a given
approximation guarantee with high probability, e.g., in the case of compressed sensing, the guarantee
would be that $f(Ax)$ satisfies the $\ell_2/\ell_2$ guarantee above, while in the case of data streams,
the guarantee may be that $f(Ax) = (1 \pm \eps)\|x\|_p$. The adversary only sees values $f(Ax^1), f(Ax^2), 
\ldots, f(Ax^t)$ for a sequence of vectors $x^1, \ldots, x^t$ of her choice, where $x^i$ may depend
on $x^1, \ldots, x^{i-1}$ and $f(Ax^1), \ldots, f(Ax^{i-1})$. The goal of the adversary is to find a vector
$x$ for which $f(Ax)$ does not satisfy the approximation guarantee. This corresponds to the private model
of compressed sensing, given in Definition 3 of \cite{ghsww12}. 

\subsection{Our Results}
We resolve the above open question in the negative. In fact, we prove a much
more general result about linear sketches. All of our results are derived from
the following promise problem {\sc GapNorm}($B$): for an input vector $x \in
\mathbb{R}^n$, output $0$ if $\|x\|_2 \leq 1$ and output $1$ if $\|x\|_2 \geq
B$, where $B > 1$ is a parameter. If $x$ satisfies neither of these two
conditions, the output of the algorithm is allowed to be $0$ or $1$. 

Our main theorem is stated informally as follows.
\begin{theorem}[Informal version of \theoremref{constant}]
There is a randomized algorithm which, given a parameter $B\ge2$ and
oracle access to a linear sketch that uses at most $r=n-O(\log(nB))$
rows, with high probability finds a distribution over queries on which 
the linear sketch fails to solve {\sc GapNorm}($B$) with constant probability.

The algorithm makes at most $\poly(rB)$ adaptively chosen queries to the
oracle and runs in time $\poly(rB).$ Moreover, the algorithm uses only $r$
``rounds of adaptivity'' in that the query sequence can be partitioned into at
most $r$ sequences of non-adaptive queries.
\end{theorem}
Note that the algorithm in our theorem succeeds on every linear sketch with
high probability. In particular, our theorem implies that one cannot design a
distribution over sketching matrices with at most $r$ rows so as to output a
value in the range $[\|x\|_2, B \|x\|_2]$, that is, a $B$-approximation to
$\|x\|_2$, and be correct with constant probability on an adaptively chosen
sequence of $\poly(rB)$ queries.  This is unless the number $r$ of rows in the
sketch is $n - O(\log(nB))$, which agrees with the trivial $r = n$ upper bound
up to a low order term.  Here $B$ can be any arbitrary approximation factor
that is only required to be polynomially bounded in $n$ (as otherwise the
running time would not be polynomial). An interesting aspect of our algorithm
is that it makes arguably very natural queries as they are all drawn from
Gaussian distributions with varying covariance structure.

We also note that the second part of our theorem implies that the queries
can be grouped into fewer than $r$ rounds, where in each round the
queries made are independent of each other conditioned on previous rounds.
This is close to optimal, as if $o(r/\log r)$ rounds were used, the sketching
algorithm could partition the rows of $A$ into $o(r/\log r)$ disjoint blocks
of $\omega(\log r)$ coordinates, and use the $i$-th block alone to respond to
queries in the $i$-th round. If the rows of $A$ were i.i.d.~normal random
variables, one can show that this would require a super-polynomial (in $r$)
number of non-adaptive queries to break, even for constant $B$.  Moreover, our
theorem gives an algorithm with time complexity polynomial in $r$ and $B$, and
therefore rules out the possibility of using cryptographic techniques secure
against polynomial time algorithms. 

We state our results in terms of algorithms that output any
computationally unbounded but deterministic function $f$ of the sketch $Ax.$
However, it is not difficult to extend all of our results to the setting where
the algorithm can use additional internal randomness at each step to output a 
randomized function $f$ of $Ax.$ This is discussed in \sectionref{rand}. 

\paragraph{Applications.}
We next discuss several implications of our main theorem.
Our algorithm in fact uses only query vectors~$x$ which are 
$O(r)$-dimensional for $B \leq \exp(r).$ Recall that for such vectors,
$\Omega(r^{-1/2}\|x\|_2)\le \|x\|_p \le O(r^{1/2}\|x\|_2),$
for all $1 \leq p \leq \infty.$ This gives us the following
corollary for any $\ell_p$-norm.
\begin{corollary}[Informal]
No linear sketch with $n-\omega(\log n)$ rows approximates the $\ell_p$-norm 
to within a fixed polynomial factor on a sequence of polynomially many 
adaptively chosen queries.
\end{corollary}
The corollary also applies to other problems that are as least
as hard as $\ell_p$-norm estimation,
such as the earthmover distance, 
or that can be embedded into $\ell_p$ with small distortion. 

Via a reduction to {\sc GapNorm}$(B)$, we are able to resolve the
aforementioned open question for sparse recovery, even when $k = 1$.  
\begin{corollary}[Informal]
Let $C\ge 1.$ No linear sketch with $o(n/C^2)$ rows guarantees
$\ell_2/\ell_2$-recovery on a polynomial number of adaptively chosen inputs. 
More precisely, we can find with probability $2/3$ an input $x$ for which the
output~$x'$ of the sketch does not satisfy $\|x-x'\|_2\le
C\|x_{\mathrm{tail}(1)}\|_2.$
\end{corollary}
For constant approximation factors $C$, this shows one cannot do asymptotically
better than storing the entire input. For larger approximation factors $C$, 
the dependence of the number of rows on $C$ in this corollary 
is essentially best possible (at least for small $k$), 
as we point out in \sectionref{recovery}.

\paragraph{Connection to Differential Privacy.}

How might one design algorithms that are robust to adversarial inputs? An
intriguing approach is offered by the notion of differential
privacy~\cite{DworkMNS06}. Indeed, differential privacy is designed to guard a
private database $D\in\bits^n$ (here thought of as $n$ private bits) against
adversarial and possibly adaptive queries from a data analyst. Intuitively
speaking, differential privacy prevents an attacker from reconstructing the private
bit string. In our setting we can think of $D$ as the random string that
encodes the matrix used by the sketching algorithm and indeed our algorithm is
precisely a \emph{reconstruction attack} in the terminology of differential privacy.  It
is known that if $D$ is chosen uniformly at random, then after
conditioning~$D$ on the output of an $\epsilon$-differentially private
algorithm, the string $D$ is a \emph{strongly
$2\epsilon$-unpredictable}\footnote{This means that each bit of $D$ is at most
$2\epsilon$-biased conditioned on the remaining bits.} random
string~\cite{McGregorMPRTV10}. Hence, if the answers given by the sketching
algorithm satisfy differential privacy, then the attacker cannot learn the
randomness used by the sketching algorithm. This could then be used to argue
that the sketch continues to be correct.

An interesting corollary of our work is that it rules out the possibility of
correctly answering a polynomial number of ``{\sc GapNorm} queries'' using the
differential privacy approach outlined above. This stands in sharp contrast to
work in differential privacy which shows that a nearly exponential number of
adaptive and adversarial ``counting queries'' can be answered while satisfying
differential privacy~\cite{RothR10,HardtR10}. A similar (though quantitatively
sharper) separation was recently shown for the stateless
mechanisms~\cite{DworkNV12} answering counting queries. While linear sketches
are stateless, the model we use here in principle permits more flexibility in
how the randomness is used by the algorithm so that the previous separation
does not apply.

\subsection{Comparison to Previous Work}
While the above papers \cite{ghrsw12,ghsww12} introduce the problem, the results they obtain do not directly address the
general problem. 
The main result of \cite{ghsww12} is that in the private model of compressed sensing, the $\ell_2/\ell_2$
error guarantee is achievable with $r = k \cdot \poly(\log n)$ measurements, under the assumption that
the algorithm has access to the exact value of $\|x\|_2^2$ as well as specific Fourier
coefficients of $x$ (or approximate values to these quantities that 
come from a distribution that depends only on the exact values). While in some applications this may be possible, 
it is not hard to show that this assumption cannot be realized by any 
linear sketch unless $r \geq n$ (nor by any low space streaming algorithm or low communication
protocol). The main result in \cite{ghrsw12} relevant to this problem is that if the adversary
can read the sketching matrix $A$, but is required to stream through the entries in a single pass using 
logarithmic space, then it cannot generate a query $x$ for which the output of the algorithm on $Ax$ does
not satisfy the $\ell_2/\ell_2$ error guarantee. This is quite different from the problem considered here, 
since we consider multiple adaptively
chosen queries rather than a single query, and we do not allow direct access to $A$ but rather only observe $A$
through the outputs $f(Ax^i)$. 

We note that other work has observed the danger of using the output $f(Ax)$ to create an input $x'$ for which the 
value $f(Ax')$ is used \cite{agm12,agm12b}. 
Their solution to this problem is just to use a new sketching matrix $A'$
drawn from $\pi$, and instead query $f(A'x')$. As mentioned, it may not be possible to do this, e.g.,
if $x'$ is a perturbation to $x$, one would need to compute $A'x'$ without knowing $x'$ (since $x$ may only be
known through the sketch $Ax$). Other work \cite{ipw11,pw12} has also considered the power of adaptively choosing 
matrices to achieve fewer measurements in compressed sensing; this is orthogonal to our work since we
consider adaptively chosen inputs rather than adaptively chosen sketches. 

Sketching in adversarial environments was also the motivation
for~\cite{MironovNS08}. However, they consider an adversarial multi-party
model that is different from ours.
\subsection{Our Techniques and Proof Overview} 
We prove our main theorem by considering the following game between two
parties, Alice and Bob. Alice chooses an $r \times n$ matrix $A$ from
distribution $\pi$. Bob makes a sequence of queries $x^1, \ldots, x^s\in\R^n$ to
Alice, who only sees $Ax^i$ on query $i$. Alice responds by telling Bob the
value $f(Ax^i)$. We stress that here $f$ is an arbitrary function here that
need not be efficiently computable, but for now we assume that $f$ uses no
randomness. This restriction can be removed easily as we show later.  Bob's
goal is to {\it learn} the row space $R(A)$ of Alice, namely the at most
$r$-dimensional subspace of $\mathbb{R}^n$ spanned by the rows of $A$.  If Bob
knew $R(A)$, he could, with probability $1/2$ query $0^n$ and with probability
$1/2$ query a vector in the kernel of $A$. Since Alice cannot distinguish the
two cases, and since the norm in one case is $0$ and in the other case
non-zero, she cannot provide a relative error approximation. Our main theorem
gives an algorithm (which can be executed efficiently by Bob) that learns
$r-O(1)$ orthonormal vectors that are almost contained in $R(A).$ While this
does not give Bob a vector in the kernel of $A,$ it effectively reduces
Alice's row space to be constant dimensional thus forcing her to make a
mistake on sufficiently many queries.

\paragraph{The conditional expectation lemma.}
In order to learn $R(A)$, Bob's initial query is drawn from the multivariate
normal distribution $N(0, \tau I_n)$, where $\tau I_n$ is the covariance
matrix, which is just a scalar $\tau$ times the identity matrix $I_n$.  This
ensures that Alice's view of Bob's query $x$, namely, the projection $P_Ax$ of
$x$ onto $R(A)$, is spherically symmetric, and so only depends on
$\|P_Ax\|_2$. Given $\|P_Ax\|_2$, Alice needs to output $0$ or $1$ depending
on what she thinks the norm of $x$ is. The intuition is that since Alice has a
proper subspace of $\mathbb{R}^n$, she will be confused into thinking $x$ has
larger norm than it does when $\|P_Ax\|_2$ is slightly larger than its
expectation (for a given $\tau$), that is, when $x$ has a non-trivial
correlation with $R(A)$. Formally, we can prove a conditional expectation
lemma showing that there exists a choice of $\tau$ for which $\E_{x \sim
N(0, \tau \mathrm{Id}_r)}\left[\|P_Ax\|_2^2 \mid f(Ax) = 1\right] - \E_{x \sim N(0, \tau
\mathrm{Id}_r)}\left[\|P_Ax\|_2^2\right]$ is non-trivially large. This is done by showing that the sum of this
difference over all possible $\tau$ in a range $[1, B]$ is noticeably
positive. Here $B$ is the approximation factor that we tolerate.
In particular, there exists a $\tau$ for which this difference is
large. To show the sum is large, for each possible condition $v =
\|P_Ax\|_2^2$, there is a probability $q(v)$ that the algorithm outputs $1$,
and as we range over all $\tau$, $q(v)$ contributes both positively and
negatively to the above difference based on $v$'s weight in the $\chi^2$-distribution 
with mean $r \cdot \tau$. The overall contribution of $v$ can be shown to be
zero. Moreover, by correctness of the sketch, $q(v)$ must typically be close to $0$ for
small values of $v,$ and typically close to $1$ for large values of $v.$ 
Therefore $q(v)$ zeros out some of the negative contributions that $v$ would otherwise
make and ensures some positive contributions in total.

\paragraph{Boosting a small correlation.}
Given the conditional expectation lemma we we can find many independently
chosen $x^i$ for which each $x^i$ has a slightly increased expected projection
onto Alice's space $R(A)$. At this point, however, it is not clear how to
proceed unless we can aggregate these slight correlations into a single vector
which has very high correlation with $R(A).$ We accomplish this by arranging
all $m=\poly(n)$ positively labeled vectors $x^i$ into an $m\times n$ matrix $G$ and
computing the top right singular vector~$v^*$ of $G.$ Note that this can be
done efficiently.  We show that, indeed, $\|P_Av^*\|\ge 1-1/\poly(n).$ In
other words $v^*$ is almost entirely contained in $R(A).$ This step is crucial
as it gives us a way to effectively reduce the dimension of Alice's space
by~$1$ as we will see next.

\paragraph{Iterating the attack.}
After finding one vector inside Alice's space, we are unfortunately not done. In fact 
Alice might initially use only a small fraction of her rows and switch to a new set
of rows after Bob learned her initial rows.
We thus iterate the previously described attack as follows. Bob now makes
queries from a multivariate normal distribution inside of the subspace
orthogonal to the the previously found vector. In this way we have effectively
reduced the dimension of Alice's space by $1$, and we can repeat the attack
until her space is of constant dimension, at which point a standard
non-adaptive attack is enough to break the sketch. Several complications
arise at this point. For example, each vector that we find is only approximately contained
in $R(A).$ We need to rule out that this approximation error could help Alice.
We do so by adding a sufficient amount of global Gaussian noise to our query
distribution. This has the effect of making the distribution statistically
indistinguishable from a query distribution defined by vectors that are exactly
contained in Alice's space. Of course, we then also need a generalized
conditional expectation lemma for such distributions.

\subsection*{Paper Outline}
We start with some technical preliminaries in \sectionref{prelims}.
We then prove the conditional expectation lemma in \sectionref{conditional}. The
proof of this lemma requires rather detailed information about averages of 
$\chi^2$-distributions in certain intervals. The development of these
bounds is contained in \sectionref{chi-squared}. 
In \sectionref{attack} we present and analyze our 
complete adaptive attack. The proof again requires several technical
ingredients. One tool (given in \sectionref{subspaces}) 
relates a distance function between two subspaces to the
statistical distance of certain distributions that we use in our attack. 
The other tool in
\sectionref{net} analyzes the top singular vector of certain biased Gaussian
matrices arising in our attack.
In \sectionref{applications} we give applications
to compressed sensing, data streams, and distributed functional monitoring. 

\section{Preliminaries}
\sectionlabel{prelims}

\paragraph{Notation.}
Given a subspace $V\subseteq\R^n,$ we denote by $P_V$ the orthogonal
projection operator onto the space~$V.$ The orthogonal complement of a linear
space~$V$ is denoted by $V^\bot.$
When $X$ is a distribution we use $x\sim X$ to indicate that $x$ is a random
variable drawn according to the distribution~$X.$

\paragraph{Linear Sketches.}
A linear sketch is given by a distribution~$\cM$ over $r\times n$ matrices 
and an evaluation mapping $F\colon\R^{r\times n}\times \R^r\to R$ where $R$
is some output space which we typically choose to be $R=\bits.$ The
algorithm initially samples a matrix $A\sim\cM.$ The answer to each
query $x\in\R^n$ is then given by $F(A,Ax).$ Since the evaluation map $F$ is not 
restricted in any way, the concrete representation of $A$ as a matrix is not
important. We will therefore identify $A$ with its image, an $r$-dimensional subspace of
$\R^n$ (w.l.o.g. $A$ has full row rank). In this case, we can write an instance of
a sketch as a mapping $f\colon\R^n\to R$ satisfying the identity $f(x) =
f(P_Ax).$ In this case we may write $f\colon A\to\bits$ even though $f$ is
defined on all of $\R^n$ via orthogonal projection onto~$A.$

\paragraph{Distributions.}
We denote the $d$-dimensional Gaussian distribution with mean $\mu\in\R^d$ and
independent coordinates with variance $\sigma^2\in\R$ by $N(\mu,\sigma^2)^d.$  
The statistical distance (or total variation distance) between two
distributions $X,Y$ is denoted by $\|X-Y\|_\tv.$

\section{Certain Averages of $\chi^2$-distributions}
\sectionlabel{chi-squared}

In this section we develop the main technical ingredients for our conditional
expectation lemma. Specifically, we will work in $\R^d$ and consider weighted 
averages of the $\chi^2$-distribution in certain intervals. 
The density function of the squared Euclidean norm of a $d$-dimensional
standard Gaussian variable is given by $\nu(s)=
s^{d/2-1}e^{-s/2}/2^{d/2}\Gamma(d/2).$ 
We let
$\nu_{\tau,d}\colon[0,\infty)\to[0,1]$ be the density function of a $\chi^2$-distribution with
$d$-degrees of freedom and expectation $\tau.$ Note that this coincides with
the density function of the squared norm of a
$d$-dimensional Gaussian variable $N(0,\tau/d)^d$ which we will denote as:
\begin{equation}\equationlabel{nu-density}
\nu_{\tau,d}(s) 
= 
\frac{d\left(\frac{sd}{\tau}\right)^{d/2-1}e^{-\frac {sd}{2\tau}}}{\tau 2^{d/2}\Gamma(d/2)}
\mper
\end{equation}
Here we used that $\nu_{\tau,d}(s)=d\nu(sd/\tau)/\tau.$
We will omit the subscript $d$ whenever it is clear from the context.
Further, let $\GammaD_d$ denote the probability measure on $[0,\infty)$ of the
Gamma distribution given by the density $\gamma_d\colon(0,\infty)\to[0,1],$
\begin{equation}\equationlabel{gamma-density}
\gamma_d(x) = \frac{x^{d-1}e^{-x}}{\Gamma(d)}\mper
\end{equation}

\begin{lemma}\lemmalabel{taunu}
Let $0\le a \le b$ and let $s\ge 0.$ Then, 
\begin{align}
\int_a^b s\nu_\tau(s)\rd \tau
& = 
\frac{s}{1-2/d}\GammaD_{d/2-1}\left(\left[\frac {sd}{2b},\frac {sd}{2a}\right]\right) 
\\
\int_a^b \tau\nu_\tau(s)\rd \tau
& = \frac s{1-6/d+8/d^2}\cdot
\GammaD_{d/2-2}\left(\left[\frac {sd}{2b},\frac {sd}{2a}\right]\right) 
\mper
\end{align}
\end{lemma}

\begin{proof}
Applying \equationref{nu-density} and substituting
 $x=sd/2\tau,$ we have 
$\frac{\rd x}{\rd \tau} = -\frac {sd}{2\tau^2} = -\frac {2x^2}{sd}\mper$
It follows that $\frac 1\tau\rd \tau = -\frac 1x\rd x.$ Put $a'=sd/2b$ and
$b'=sd/2a.$
Thus,
\begin{align*}
\int_a^b \nu_\tau(s)\rd \tau
= \int_a^b 
\frac{d\left(\frac{sd}{\tau}\right)^{d/2-1}e^{-\frac{sd}{2\tau}}}{\tau
2^{d/2}\Gamma(d/2)}\rd \tau
= \int_{a'}^{b'} \frac{dx^{d/2-2}e^{-x}}{2\Gamma(d/2)}\rd x\mper
\end{align*}
On the other hand, $\Gamma(d/2)=(d/2-1)\Gamma(d/2-1)\mper$ Hence,
\[
\int_a^b \nu_\tau(s)\rd \tau
= \frac d{2(d/2-1)}
\GammaD_{d/2-1}\left(\left[a',b'\right]\right)
= \left(1+\frac 2{d-2}\right)
\GammaD_{d/2-1}\left(\left[a',b'\right]\right) \mper
\]
The second equation is shown similarly, again substituting $x=sd/2\tau$ and
noting that $\rd\tau = -\frac {sd}{2x^2}\rd x,$
\begin{align*}
\int_a^b \tau\nu_\tau(s)\rd \tau
= 
\int_a^b\frac{d\left(\frac{sd}{\tau}\right)^{d/2-1}e^{-\frac
{sd}{2\tau}}}
{2^{d/2}\Gamma(d/2)}\rd \tau
= \int_{a'}^{b'} \frac {sd^2}4 \cdot\frac{x^{d/2-3}e^{-x}}{\Gamma(d/2)}\rd x
= \frac {sd^2 \GammaD_{d/2-2}\left(\left[a',b'\right]\right) }
{4(d/2-1)(d/2-2)}
\mper
\end{align*}
Furthermore, 
\[
\frac{d^2}{4(d/2-1)(d/2-2)}
= \frac 1{4(1/4-1/2d-1/d+2/d^2)}
= \frac 1{(1-6/d+8/d^2)}
\]
\end{proof}

Let us introduce the function $\Delta\colon[0,\infty)\to\R_{\ge0},$ defined as
\begin{equation}
\Delta(s)
\defeq\int_d^{Bd} (s-\tau)\nu_\tau(s)\rd\tau\mper
\end{equation}
Here $B>4$ is some parameter that we will choose later.
\figureref{delta} illustrates the behavior of this function.
\begin{figure}[ht]
\begin{center}
\includegraphics{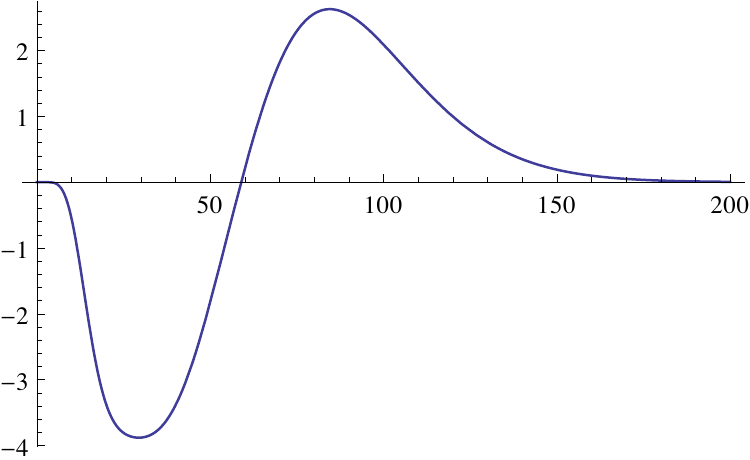}
\end{center}
\caption{$\Delta(s)$ plotted for $d=20$ and $B=4.$}
\figurelabel{delta}
\end{figure}
The next lemma states the properties of $\Delta$ that we will need.
\begin{lemma}\lemmalabel{negative}
Assume $d\ge 20.$
Then, for every $s\in[0,Bd/2],$ we have that $\Delta(s)<0.$
Moreover, for every $s\in[d,2d],$ we have $\Delta(s)<-s/3d.$
\end{lemma}
\begin{proof}
First consider the case where $s\in [2d,Bd/2].$
By \lemmaref{taunu}, we have
\[
\Delta(s)
=
\frac{s}{1-2/d}\GammaD_{d/2-1}\left(\left[\frac {s}{2B},\frac {s}{2}\right]\right) 
- \frac s{1-6/d+8/d^2}\cdot
\GammaD_{d/2-2}\left(\left[\frac {s}{2B},\frac {s}{2}\right]\right)
\]
On the other hand, for this choice of $s,$ we have
\[
\GammaD_{d/2-2}\left(\left[\frac {s}{2B},\frac {s}{2}\right]\right)
\ge
\GammaD_{d/2-2}\left(\left[\frac {d}{4},d\right]\right)
\ge 1-1/d
\]
Here, the last step can be verified directly by using that $\GammaD_{d/2-2}$
is strongly concentrated around its mean $d/2-2$ and has variance bounded by
$\sqrt{d}.$ For $d\ge 20,$ the approximation we used is valid.
Hence,
\[
\Delta(s) \le s
\left( \frac{1}{1-2/d} 
- \frac {1-1/d}{1-6/d+8/d^2}\right)
= a\left( \frac{1-1/d}{1-3/d+2/d^2} 
- \frac {1-1/d}{1-6/d+8/d^2}\right)
< -\frac sd\mper
\]
Here we used our lower bound on $d$ again.

Now consider the case $s\in[d-4,2d].$ In this case we have
$\GammaD_{d/2-2}([s/2B,s/2])\ge \GammaD_{d/2-2}([d/B,d/2-2])\ge 1/2-1/d$ by
concentration bounds for $\GammaD$ and using that the median of
$\GammaD_{d/2-2}$ is at most $d/2-2.$ For the same reason,
$\GammaD_{d/2-2}([s/2B,s/2])\ge 1/2-1/d.$ Moreover, we have that
$\GammaD_{d/2-2}([s/2B,s/2])\ge \GammaD_{d/2-1}([s/2B,s/2])$ because
$\GammaD_{d/2-1}([d/2-2,\infty))\ge\GammaD_{d/2-2}([d/2-2,\infty]).$ This follows
because $\GammaD_{d/2-1}$ has larger mean and greater variance than
$\GammaD_{d/2-2}.$ Hence,
\[
\Delta(s)\ge s(1/2-1/d)
\left( \frac{1}{1-2/d} 
- \frac {1}{1-6/d+8/d^2}\right)\le - \frac s{2d}+\frac s{d^2}
\le -\frac s{3d}\mper
\]

Finally, let $s\in[0,d-4].$ In this case we have
$[s/2B,s/2]\subseteq[0,d/2-2].$ 
But, for every $x\in[0,d/2-2],$ we have 
\[
\gamma_{d/2-2}(x)
= \gamma_{d/2-1}(x)\left(\frac{d/2-2}{x}\right)
\ge \gamma_{d/2-1}(x)\mper
\]
Hence, $\Delta(s)< 0.$

\end{proof}

Our main lemma in this section is stated next.
\begin{lemma}
\lemmalabel{s-tau}
Let $d\ge d_0$ for a sufficiently large constant $d_0.$ Let $B>4.$
Let $h\colon[0,\infty)\to[0,1]$ be any 
function satisfying the properties: 
\begin{enumerate}
\item\itemlabel{u-lower} $\int_{Bd/2}^{2Bd} (1-h(s))\rd s\le 1/Bd,$
\item\itemlabel{l-upper} $\int_0^{2d}h(s)\rd s\le 1/d.$
\end{enumerate}
Then, we have
\begin{equation}
\int_{s=0}^\infty\int_{\tau=l}^u (s-\tau)\nu_\tau(s)h(s)\rd \tau\rd s \ge
\frac d4
\mper
\end{equation}
\end{lemma}

\begin{proof}
First observe that
\[
\int_{s=0}^\infty\int_{\tau=l}^u (s-\tau)\nu_\tau(s)h(s)\rd \tau\rd s
= \int_{s=0}^\infty h(s)\Delta(s)\rd s
\]
Moreover, since $\int_0^\infty s\nu_\tau(s)\rd s = \tau,$ we have
\begin{equation}\equationlabel{zero}
\int_{s=0}^\infty \Delta(s) = 0\mper
\end{equation}
Let us consider the three intervals 
\[
L=[0,2d),\quad M=[2d,Bd/2),\quad U=[t,\infty).
\]
\begin{claim}
Without loss of generality,
$\int_M h(s)\Delta(s)\rd s \ge \int_M \Delta(s)\rd s$
\end{claim}
\begin{proof}
\lemmaref{negative} tells us that $\Delta(s)<0$ for all  $s\in M.$ Since we're
interested in lower bounding $\int_{s=0}^\infty h(s)\Delta(s),$ we can
therefore assume without loss of generality that $h(s)=1$ for all $s\in M.$ 
\end{proof}
\begin{claim}
$\int_U h(s)\Delta(s)\rd s \ge \int_U \Delta(s)\rd s - 6\mper$
\end{claim}
\begin{proof}
The claim follows from the first condition on $h$ which implies that
$h(s)=1$ almost everywhere in
the interval $I=[Bd/2,2Bd].$ In particular,
\begin{align*}
\int_I h(s)\Delta(s)\rd s
& = \int_I\Delta(s)\rd s + \int_I(1-h(s))\Delta(s)\rd s\\
& \ge \int_I \Delta(s)\rd s - \frac1{Bd}\max_{s\in I}|\Delta(s)|
\ge \int_I \Delta(s)\rd s- 4\mper
\end{align*}
Here we used that $|\Delta(s)|\le 4Bd.$
Moreover, for every $\tau\in[d,Bd]$ we have that
$\int_{[2Bd,\infty)}\tau \nu_\tau(s)\rd s\le 1/2Bd$ by standard tail bounds for
$\nu_\tau$ and sufficiently large $d_0.$
This implies
\[
\int_U h(s)\Delta(s)\rd s \ge \int_I h(s)\Delta(s)\rd s - 1
\mper
\]
Similarly, $\int_I\Delta(s) \ge \int_U\Delta(s) - 1.$ The claim follows by
combining these statements.
\end{proof}

\begin{claim}
$\int_L h(s)\Delta(s)\rd s \ge \int_L \Delta(s) + d/3 - 4.$
\end{claim}
\begin{proof}
Here we use the second condition on the claim which implies
\[
\int_L h(s)\Delta(s)\rd s \ge -\max_L|\Delta(s)|\int_L h(s)\ge -4\mcom
\]
where we used that $|\Delta(s)|\le 4d$ in this range.
On the other hand, by \lemmaref{negative}, $\Delta(s)<0$ for all $s\in L$ and
for $s\in[d,2d]$ we have $\Delta(s)<-s/3d.$ Hence,
\[
\int_L \Delta(s)\rd s
\le -\frac1{3d}\int_{d}^{2d}s\rd s
\le - \frac d3\mper
\]
\end{proof}
Combining all three claims we get
\[
\int_0^\infty h(s)\Delta(s)\rd s
\ge \int_0^\infty \Delta(s)\rd s + \frac d3 - 10
= \frac d3 - 10 \mcom
\]
where we used \equationref{zero} in the last step. For sufficiently large $d,$ 
we have $d/3-10\ge d/4$ and the lemma follows.
\end{proof}

\section{Conditional Expectation Lemma}
\sectionlabel{conditional}

The key tool in our algorithm is what we call the \emph{conditional
expectation lemma}.
Informally, it shows that we can always find a distribution over inputs that
have a non-trivially large correlation with the unknown subspace used by the
linear sketch. Our presentation here, however, will not need the
interpretation in terms of linear sketches.
Fix a $d$-dimensional linear subspace $U\subseteq\R^n.$ Throughout this
section, we think of $d$ as being lower bounded by a sufficiently large
constant.
We will consider functions of the type
$f\colon \R^n\to\bits$ which satisfy the identity $f(x) = f(P_Ux)$ for all
$x\in\R^n.$ To indicate that this identity holds we will write $f\colon
U\to\bits.$
As explained in \sectionref{prelims}, we can think of these
functions as instances of a linear sketch. Our presentation here will not need
this fact though. 

\begin{definition}[Subspace Gaussian]
Let $U\subseteq \R^n$ be a linear subspace of $\R^n.$ 
We say that a family of distributions $\cG(U)=\{g_\tau\}_{\tau\in(0,\infty)}$ 
is a \emph{subspace Gaussian} family if
\begin{enumerate}
\item $P_Ug_\tau$ is distributed like a standard Gaussian variable inside $U$
satisfying $\E\|P_Ug_\tau\|^2=\tau.$
\item $P_{U^\bot}g_\tau$ is a spherical Gaussian distribution 
that does not depend on $\tau$ and is moreover 
statistically independent of $P_Ug_\tau.$
\end{enumerate}
\end{definition}

\begin{lemma}
\lemmalabel{suffstat}
The norm $\|P_Ug_\tau\|^2$ is a sufficient statistic for 
a subspace Gaussian family $\cG(U)=\{g_\tau\}_\tau.$
Formally, for every $s>0,$ the distribution of $g_\tau$ is independent 
of $\tau$ under the condition that $s=\|P_Ug_\tau\|^2.$
\end{lemma}

\begin{remark}
The reader concerned about the condition $\|P_Ug_\tau\|^2=s$ (which has
probability $0$ under $g_\tau$) is referred to the excellent
article of Chang and Pollard~\cite{ChangP97} (Example 6) where it is shown how to 
formally justify this conditional distribution using the notion of a \emph{disintegration}.
Specifically, here we mean that the distributions $g_\tau$ have a joint
disintegration in terms of the variable $\|P_Ug_\tau\|^2.$ As explained
in~\cite{ChangP97}, we can prove the above lemma by appealing to the
factorization theorem for sufficient statistics described therein.
\end{remark}

\begin{proof}[Proof of \lemmaref{suffstat}]
Note that $g_\tau=g_1+g_2$ where $g_2$ is some distribution independent of
$\tau$ supported on $U^\bot$ and $g_1$ is supported on $U.$
By spherical symmetry of both $g_1$ and $g_2$ we
may assume without loss of generality that $U$ is a coordinate subspace, say,
the first $d=\dim(U)$ coordinates of the standard basis. Since $g$ is
independent of $g_\tau$ and supported on a disjoint set of coordinates, it
suffices to verify the claim for $g_1.$
Specifically, by the Factorization theorem for sufficient 
statistics~(see \cite{ChangP97}), we need to
show that the density of $g_1$ can be factored into the product of two
functions such $f(x)$ and $h_\tau(x)$ that $f$ does not depend on $\tau$ and
$h_\tau(x)$ depends on $\tau$ but is a function of the parameter
$\|x\|^2.$ This follows directly from the fact that the Gaussian
density at a point $x$ depends only on $\|x\|^2.$
\end{proof}

The following definition captures the condition that $f$ should evaluate to
$1$ on inputs that have large norm and should evaluate to $0$ on
inputs that have small norm.

\begin{definition}[Soundness]
\definitionlabel{soundness}
We say that a function $f\colon U\to\bits$ is \emph{$B$-sound} for a subspace Gaussian family
$\cG(U),$ where $\dim(U)=d,$ if it satisfies the requirements:
\begin{enumerate}
\item $\int_{Bd/2}^{2Bd} \E\left[ f(g_\tau)\mid \|P_Ug_\tau\|^2=s\right]\rd s
\le 1/Bd\mper$
\item $\int_0^{2d}\E\left[ f(g_\tau)\mid \|P_Ug_\tau\|^2= s\right]\rd s\le 1/d.$
\end{enumerate}
\end{definition}

We are ready to state and prove the Conditional Expectation Lemma.

\begin{lemma}
\lemmalabel{conditional}
Let $B\ge 4.$
Let $\cG(U)$ be a subspace Gaussian family where $U$ has dimension
sufficiently large dimension $d\ge d_0.$
Suppose $f\colon U\to\bits$ is $B$-sound for $\cG(U).$ 
Then, there exists $\tau\in[d,Bd]$ such that 
\begin{enumerate}
\item $\E \left[\|P_Ug_\tau\|^2 \, \Big|\, f(g_\tau)=1\right]
\ge \E\left[\|P_Ug_\tau\|^2\right] + \frac1{4B}$
\item $\Pr\Set{f(g_\tau)=1}\ge \frac1{40B^2d}\mper$
\end{enumerate}
\end{lemma}
\begin{proof}
Define the function $h\colon(0,\infty)\to\R$ by putting
\begin{equation}
h(s) = \E\left[f(g_\tau)\mid \|P_Ug_\tau\|^2=s\right]\mper
\end{equation}
Note that this is well-defined by \lemmaref{suffstat}.
Let us first rewrite the conditional expectation as follows.
\begin{align*}
\E\left[\|P_Ug_\tau\|^2 \, \Big|\, f(g_\tau)=1\right]
& = \int_0^\infty s\Pr\Set{ \|P_Ug_\tau\|^2=s \mid f(g_\tau)=1}\rd s\\
& = \int_0^\infty s\Pr\Set{ f(g_\tau)=1 \mid \|P_Ug_\tau\|^2=s }
\cdot \frac{\nu_\tau(s)}{\Pr\Set{ f(g_\tau)=1}} \rd s
\tag{by Bayes' rule}\\
& = \int_0^\infty  \frac{sh(s)\nu_\tau(s)}{\Pr\Set{ f(g_\tau)=1}} \rd s
\end{align*}
Note that here and in the following $\nu_\tau=\nu_{\tau,d},$ i.e., the
$\chi^2$-distribution has $d$ degrees of freedom corresponding to the
dimension of $U.$
\begin{claim}
The lemma follows from
follows from the following inequality:
\begin{equation}\equationlabel{subgoal}
\int_l^u\int_0^\infty(s-\tau)\nu_\tau(s)h(s)\rd s\rd \tau
\ge \frac{d}4.
\end{equation}
\end{claim}
\begin{proof}
Indeed, assuming the above inequality, it follows that there must be a
$\tau\in[d,Bd]$ such that 
\[
\int_0^\infty s\nu_\tau(s)h(s)\rd s
\ge 
\tau \int_0^\infty \nu_\tau(s)h(s)\rd s
+ \frac{d}{4Bd}
=\tau\Pr\Set{f(g_\tau)=1} 
+ \frac 1{4B}
\mper
\]
In particular,
\[
\E\left[\|P_Ug_\tau\|^2 \, \Big|\, f(g_\tau)=1\right]
=\int_0^\infty  \frac{sh(s)\nu_\tau(s)}{\Pr\Set{ f(g_\tau)=1}} \rd s
\ge \frac{\tau\Pr\Set{f(g_\tau)=1}+1/4B}{\Pr\Set{ f(g_\tau)=1}}
\ge \tau + \frac{1/4B}{\Pr\Set{f(g_\tau)=1}}.
\]
Since $\Pr\Set{f(g_\tau)=1}\le1,$ this gives us the first
conclusion of the lemma.
It remains to lower bound $\Pr\Set{f(g_\tau)=1}.$
Here we use that $\int_{5Bd}^\infty
s\nu_\tau(s)h(s)\rd s\le \frac12\int_0^\infty s\nu_\tau(s)h(s)\rd s,$ by standard 
concentration properties of $\nu_\tau.$
Hence,
\[
\Pr\Set{f(g_\tau)=1} 
= \int_0^\infty \nu_\tau(s)h(s)\rd s
\ge \frac1{10Bd}\int_0^\infty s\nu_\tau(s)h(s)\rd s
\ge \frac{1/4B}{10Bd}
= \frac{1}{40B^2d}\mper
\]
\end{proof}
By the previous claim, it suffices to prove \equationref{subgoal}. To do so we
will apply \lemmaref{s-tau}. The lemma in fact directly implies the claim, if
we can show that $h$ satisfies the properties required in \lemmaref{s-tau}.
It is easily verified that these properties coincide with the soundness
assumption on~$f.$ Thus,
\[
\int_l^u\int_0^\infty(s-\tau)\nu_\tau(s)h(s)\rd s\rd \tau
\ge \frac d3-2\ge \frac d4\mper
\]
This concludes the proof of \lemmaref{conditional}.
\end{proof}

The following corollary is a direct consequence of \lemmaref{conditional} that
states that there is one direction in the subspace that has increased
variance.

\begin{corollary}
\corollarylabel{conditional}
Let $\cG(U)$ satisfy the assumptions of \lemmaref{conditional}. 
Then, there is $\tau\in[d,Bd]$ and a vector $u\in U,$ satisfying,
\begin{enumerate}
\item $\E \left[\langle u,g_\tau\rangle^2 \, \Big|\, f(g_\tau)=1\right]
\ge \E\left[\langle u,g_\tau\rangle^2\right] + \frac 1{4Bd}$
\item $\Pr\Set{f(g_\tau)=1}\ge \frac1{40B^2d}\mper$
\end{enumerate}
\end{corollary}
\begin{proof}
Pick an arbitrary orthonormal basis $u_1,\dots,u_d$ of $U.$ Since
$\|P_Ug_\tau\|^2=\sum_{i=1}^d\langle u_i,g_\tau\rangle^2,$ one of the basis
vectors must satisfy the conclusion of the lemma by an averaging argument.
\end{proof}

\subsection{Noisy orthogonal complements}

In this section we extend the conditional expectation lemma to a family of
distributions that will be important to us later on. We will fix a
$r$-dimensional subspace $A\subseteq\R^n.$
The family of
distributions we will define next isn't subspace Gaussian on $A,$ but rather 
subspace Gaussian on $A\cap V^\bot,$ where $V\subseteq A$ is a linear subspace
of $A$ of dimension $\dim(V)\le d-1.$

A distribution in this family is given by a subspace $V$ and a variance
$\sigma^2.$ Intuitively, the distribution corresponds to a Gaussian
distribution on the subspace $V^\bot$ of variance $\sigma^2$ 
plus a small Gaussian supported on all of $\R^n$ of constant variance
independent of $\sigma^2.$ The formal definition is given next.

\begin{definition}
\definitionlabel{complement-gaussian}
Let $\sigma>0.$
Given a subspace $V\subseteq A$ of dimension $t\le r-1$ and let $d=r-t.$ We
define the distribution $G(V^\bot,\sigma^2)$
as the distribution obtained from
sampling $g_1\sim N(0,\sigma^2)^n, g_2\sim N(0,1/4)^n$ independently outputting
$g=P_{V^\bot}g_1 + g_2.$ 

Further we define the family of distributions $\cG(A\cap V^\bot)=\{g_\tau\}$
by letting $g_\tau=P_A g $ where $g\sim G(V^\bot,\tau/d - 1/4)$ if
$\tau/d>1/4$ and otherwise we put $g_\tau = P_A  g$ where $g\sim N(0,\tau/d)^n$ otherwise.
\end{definition}

The next lemma confirms that $\cG(A\cap V^\bot)$ is subspace Gaussian.

\begin{lemma}\lemmalabel{subspace-gaussian}
$\cG(A\cap V^\bot)$ is a
subspace Gaussian family. 
\end{lemma}
\begin{proof}
Let $U=A\cap V^\bot.$ For $\tau\le 1/4,$ it is clear that
$\E\|P_Ug_\tau\|^2=d\tau/d=\tau$ as is required and $\|P_{U^\bot}g_\tau\|=0.$  
For $\tau> 1/4,$ recall that
$g=P_{V^\bot}g_1+g_2.$ Hence, inside $P_Ug$ is distributed like a spherical
Gaussian with variance $\tau/d$ in each direction. In particular, 
\[ \|P_U
g_\tau\|^2 = \|P_U (g_1 + g_2)\|^2 = d\tau/d=\tau\mper 
\] 
On the other hand, $P_{U^\bot}g$ only depends on $g_2$ and is hence
independent of $\tau.$ This shows the second property of subspace Gaussian. 
\end{proof}

The next definition captures the correctness requirement on~$f$ for inputs
drawn from the distribution $G(V^\bot,\sigma^2).$

\begin{definition}[Correctness]
\definitionlabel{correctness}
We say that a function $f\colon A\to\bits$ is \emph{$(\epsilon,B)$-correct} on $V^\bot$
with $d=\dim(V^\bot\cap A)$ if:
\begin{enumerate}
\item for all $\sigma^2\in[B/2,2B]$ and $g\sim G(V^\bot,\sigma^2)$ we have 
$\Pr\Set{f(g)=1}\ge 1-\epsilon$
\item for all $\sigma^2\in[0,2]$ and $g\sim G(V^\bot,\sigma^2)$ we have
$\Pr\Set{f(g)=1}\le \epsilon.$ 
\end{enumerate}
We say that $f$ is \emph{$B$-correct} on $V^\bot$ if it is
$(\epsilon,B)$-correct for some $\epsilon\le 1/10(Bd)^2.$ 
\end{definition}

We will now relate the correctness definition to our earlier soundness
definition.

\begin{lemma}\lemmalabel{correct2sound}
If $f$ is $B$-correct on $V^\bot,$ then $f$ is
$B$-sound for $\cG(A\cap V^\bot).$ 
\end{lemma}
\begin{proof}
We will prove the claim in its contrapositive. Indeed suppose that $f$ is not
$B$-sound for $\cG(A\cap V^\bot).$ This means that one of the two
requirements in \definitionref{soundness} is not satisfied. Suppose it is the
first one. In this case we know that for $I=[Bd/2,2Bd]$
and $h(s)=\Pr\Set{f(g_\tau)=1\mid \|g_\tau\|^2=s},$ we have 
$\E_{s\in I} (1-h(s)) > 1/2(Bd)^2.$ Suppose we sample $g\sim G(V^\bot,\sigma^2)$
where $\sigma^2$ is chosen uniformly at random from $B/2,2B.$ We claim that
that the distribution of $\|g\|^2$ is pointwise within a factor~$5$ of the uniform
distribution inside the interval $[B/2,2B].$ Hence, $\E h(\|g\|^2) >
1/10(Bd)^2.$ This violates the first condition of correctness. 

The case where the second requirement of \definitionref{soundness} is violated
follows from an analogous argument.
\end{proof}

Below we state a variant of the conditional expectation lemma for
distributions of the above form. Moreover, we will remove the requirement that
$t\le r-d_0$ and obtain a result that applies to any $t\le \dim(A).$

\begin{lemma}\lemmalabel{conditional-main}
Let $A\subseteq\R^n$ be a subspace of dimension $\dim(A)=r\le n-d_0$ for some
sufficiently large constant $d_0.$
Let $V\subseteq A$ be a subspace of $A$ of dimension $t\le r.$ 
Suppose that $f\colon A\to\bits$ is $(1/10(d_0B)^2,B)$-correct on $V^\bot.$
Then, there exists a scalar $\sigma^2\in[3/4,B],$ and a vector $u\in
A\cap V^\bot$ such that for $g\sim G(V^\bot,\sigma^2)$ we have for
$d=\max\{r-d,d_0\}:$
\begin{enumerate}
\item $\E \left[\langle u,g\rangle^2 \, \Big|\, f(g)=1\right]
\ge \E\left[\langle u,g\rangle^2\right] + \frac1{4Bd}$
\item $\Pr\Set{f(g)=1}\ge \frac1{40B^2d}\mper$
\end{enumerate}
\end{lemma}

\begin{proof}
Before we proceed we would like to ensure that $\dim(A\cap V^\bot)$ is at
least a sufficiently large constant $d_0.$ This can be ensured without loss of 
generality by considering instead a subspace $A'\supseteq A$ of dimension $r+d_0$
obtained by extending $A$ arbitrarily to $r+d_0$ dimensions. This can be done
since $n\ge r+d_0.$ Define the function $f'(x)=f(P_Ax).$ Note that $f'(x)=f(x)$ 
on all $x\in\R^n.$ Hence, $f'$ is still $(1/10(d_0B)^2,B)$-correct on
$V^\bot.$ Moreover, now $\dim(V^\bot\cap A)=d_0.$ Hence, by
\lemmaref{correct2sound}, we have that $f$ is sound for the subspace Gaussian
family $\cG(A'\cap V^\bot).$ Let $U=A'\cap V^\bot.$
We can apply \corollaryref{conditional} to $\cG(U)$ to conclude
that there is $\tau\in[d,Bd]$ and $u\in U$ such that
\[
\E \left[\langle u,g_\tau\rangle^2 \, \Big|\, f'(g_\tau)=1\right]
\ge \E\left[\langle u,g_\tau\rangle^2\right] + \frac1{4Bd}
\]
and $\Pr\Set{f'(g_\tau)=1}\ge \frac1{40B^2d}\mper$ By definition of $f'$ the
condition $f'(g_\tau)=1$ is equivalent to $f(g_\tau)=1.$ 
The condition $f(g_\tau)=1$ does not affect any vector that is orthogonal to
$A.$ Hence we may assume that $u\in A\cap V^\bot.$ Also note that
$g_\tau=P_{A'}g$ for some $g\sim G(V^\bot,\sigma^2)$ with
$\sigma^2\in[3/4,B].$ Moreover, $f(g)=f'(g_\tau),$ and also
$\langle u,g_\tau\rangle = \langle u,g\rangle$ since $u\in U.$
Hence, we have
\[
\E \left[\langle u,g\rangle^2 \, \Big|\, f(g)=1\right]
\ge \E\left[\langle u,g\rangle^2\right] + \frac1{4Bd}
\]
with $\Pr\Set{f(g)=1}\ge \frac1{40B^2d}\mper$ This is what we wanted to
show.
\end{proof}

\subsection{Distance between subspaces}
\sectionlabel{subspaces}

Our goal is to relate distributions of the form $G(V^\bot,\sigma^2)$ to
$G(W^\bot,\sigma^2)$ where $V$ and $W$ are subspaces. For this purpose, 
we consider the following distance function $d(V,W)$ between two subspaces
$V,W\subseteq\R^n:$
\begin{equation}
d(V,W)=\|P_V-P_W\|_2=\sup_{v\in\R^n}\frac{\|P_Vv-P_Wv\|}{\|v\|}\mper
\end{equation}
We will show that if $V$ and $W$ are close in this distance measure, then the
two distributions $G(V^\bot,\sigma^2)$ and $G(W^\bot,\sigma^2)$ are
statistically close. Recall that we denote the statistical distance between
two distributions $X,Y$ by $\|X-Y\|_\tv.$ We need the following well-known
fact.
\begin{fact}\factlabel{shifted-gaussians}
Let $v\in\R^n.$ Then,
\[
\|N(0,\sigma^2)^n-N(v,\sigma^2)^n\|_\tv \le \frac{\|v\|}{\sigma}\mper
\]
\end{fact}
Using this fact we can express the statistical distance between
$G(V^\bot,\sigma^2)$ and $G(W^\bot,\sigma^2)$ for two subspaces $V,W$ in
terms of the distance $d(V,W).$
\begin{lemma}
\lemmalabel{stat-dist}
For every $\sigma^2\in(0,B],$ we
have
\[
\|G(V^\bot,\sigma^2) - G(W^\bot,\sigma^2)\|_\tv 
\le 20\sqrt{Bn\log(Bn)}\cdot d(V,W)+\frac1{(Bn)^5}\mper
\]
\end{lemma}
\begin{proof}
Sample $g_1\sim N(0,\sigma^2)^n$ and $g_2,g_2'\sim N(0,1/4)^n$ independently.
Let us denote by $x = P_{V^\bot}g_1 + g_2$ and by $y=P_{W^\bot}g_1 + g_2'.$
Note that $x$ is distributed like a random draw from $G(V^\bot,\sigma^2)$ and
$y$ like a draw from $G(W^\bot,\sigma^2).$ However, we introduced a
dependence throw $g_1.$ Note that it is sufficient to bound the statistical
distance of these coupled variables. On the one hand,
\[
\| P_{V^\bot}g_1 - P_{W^\bot} g_1 \|
= \| P_{V}g_1 - P_{W} g_1 \|
\le \|g_1\|\cdot d(V,W)
\]
On the other hand, by Gaussian concentration bounds,
\[
\Pr\Set{ \|g_1\|\le 10\sqrt{Bn\log(Bn)} } \le \frac 1{(Bn)^5}\mper
\]
Condition on $\|g_1\|\le 10\sqrt{Bn\log(Bn)}.$
Under this condition,
for every possible value $u=P_{V^\bot}g_1 -P_{W^\bot}g_1,$ we have
\begin{align*}
\|N(u,1/4)^n - N(0,1/4)^n\|_\tv 
& \le 2\|u\| \tag{by \factref{shifted-gaussians}} \\
& \le 2\|g_1\|\cdot d(V,W)
\le 20\sqrt{Bn\log(Bn)}\cdot d(V,W)\mper
\end{align*}
Noting that $u + N(0,1/4)^n = N(u,1/4)^n,$ it follows
\begin{align*}
\|P_V^{\bot}g_1 + N(0,1/4)^n - P_W^{\bot}g_1 +  N(0,1/4)^n\|_\tv
&= \|N(u,1/4)^n - N(0,1/4)^n\|_\tv \\
&\le 20\sqrt{Bn\log(Bn)}\cdot d(V,W)\mper
\end{align*}
Finally, since the condition $\|g_1\|\le 10\sqrt{Bn\log(Bn)}$ has probability
$1-1/(Bn)^5,$ removing it can only increase the statistical distance of the two
variables by additive $1/(Bn)^5.$ 
\end{proof}

\section{An Adaptive Reconstruction Attack}
\sectionlabel{attack}

We next state and prove our main theorem.
It shows that no function $f\colon\R^n\to\bits$ that depends
only on a lower dimensional subspace can correctly predict the $\ell_2^2$-norm
up to a factor~$B$ on a polynomial number of adaptively chosen inputs. Here,
$B$ can be any factor and the complexity of our attack will depend on~$B$ and
the dimension of the subspace.  We will in fact show a more powerful
distributional result. This result states that no such function can predict
the $\ell_2^2$-norm on a rather natural sequence of distributions even if we
allow the function to err on each distribution with inverse polynomial
probability. This distributional strengthening will be useful in our
application to compressed sensing later on.
The next definition formalizes the way in which a linear sketch will fail
under our attack.

\begin{definition}[Failure certificate]
Let $B\ge 8$ and let $f\colon\R^n\to\bits.$ We say that a pair $(V,\sigma^2)$
is a \emph{$d$-dimensional failure certificate for $f$} if $V\subseteq\R^n$ is $d$-dimensional
subspace and $\sigma^2\in[0,2B]$ such that for some constant $C>0,$ 
we have $n\ge d+10C\log(Bn)$ and moreover:
\begin{itemize}
\item Either $\sigma^2\in[B/2,50B]$ and 
$\Pr_{g\sim G(V^\bot,\sigma^2)}\Set{f(g)=1}\le 1- (Bn)^{-C},$
\item or $\sigma^2\le 2$ and
$\Pr_{g\sim G(V^\bot,\sigma^2)}\Set{f(g)=1}\ge n^{-C}.$
\end{itemize}
\end{definition}
The motivation for the previous definition is given by the next simple fact
showing that a failure certificate always gives rise to a distribution over
which $f$ does not decide the {\sc GapNorm} problem up to a factor~$\Omega(B)$
on a polynomial number of queries. We note that in \sectionref{high-error} we
strengthen this concept to give a distribution where~$f$ errs with constant
probability.
\begin{fact}
Given a $d$-dimensional failure certificate for~$f,$ we can find with
$\poly(Bn)$ non-adaptive queries with probability $2/3$ an
input $x$ such that either $\|x\|^2\ge B(n-d)/3$ and $f(x)=0$ or $\|x^2\|\le
3(n-d)$ and $f(x)=1.$
\end{fact}
\begin{proof}
Sample $O((Bn)^C)$ queries from $G(V^\bot,\sigma^2).$ Suppose
$\sigma^2\le 2.$ Since $n-d$ is
sufficiently large compared to $d,$ by a union bound and Gaussian
concentration, we have that with high probability simultaneously for all queries $x,$
$\|x\|^2\le 3(n-d).$ On the other hand, with high probability, $f$ outputs~$1$
on one of the queries. The case where $\sigma^2\ge B/2$ follows with the
analogous argument.
\end{proof}

Our next theorem shows that we can always find a failure certificate with a
polynomial number of queries. 

\begin{theorem}[Main]\theoremlabel{attack}
Let $B\ge 8.$ Let $A\subseteq\R^n$ be a $r$-dimensional subspace of $\R^n$
such that $n\ge r+90\log(Br).$ Assume that $B\le\poly(n).$ Let $f\colon \R^n\to\bits$ satisfying
$f(x)=f(P_A x)$ for all $x\in\R^n.$ Then, there is an algorithm that given
only oracle access to $f$ finds with probability $9/10$ a failure certificate for $f.$ 
The time and query complexity of the algorithm is bounded by $\poly(B,r).$
Moreover, all queries that the algorithm makes are sampled from
$G(V^\bot,\sigma^2)$ for some $V\subseteq\R^n$ and $\sigma^2\in(0,B].$ 
\end{theorem}
We next describe the algorithm promised in \theoremref{attack} in
\figureref{attack}.

\begin{figure}[h]
\begin{boxedminipage}{\textwidth}
\noindent \textbf{Input:} Oracle~$\cA$ providing access to a function
$f\colon\R^n\to\bits,$ parameter $B\ge 4.$

\medskip
\noindent \textbf{Attack:}
%
Let $V_1 = \{0\}, m=O(B^{13}n^{11}\log^{15}(n)),$ and $S=[3/4,B]\cap
\epsilon\mathbb{Z}$ where $\epsilon=1/20(Bn)^2\log(Bn).$ 

\medskip
\noindent {\bf For} $t=1$ {\bf to} $t=r+1:$
\begin{enumerate}
\item {\bf For each} $\sigma^2\in S:$
\begin{enumerate}
\item Sample $g_1,\dots,g_m\sim G(V_t^\bot,\sigma^2).$ Query $\cA$ on each $g_i.$ Let $a_i=\cA(g_i).$
\item\itemlabel{empirical} Let $s(t,\sigma^2)=\frac1m\sum_{i=1}^m a_i$ 
denote the fraction of samples that are positively labeled. 
\begin{itemize}
\item
If either $\sigma^2\ge B/2$ and $s(t,\sigma^2)\le 1-\epsilon,$ or, $\sigma^2\le
2$ and $s(t,\sigma^2)\ge\epsilon,$ then
terminate and {\bf output} $(V_t^\bot,\sigma^2)$ as a purported failure certificate.
\item Else let $g_1',\dots,g_{m'}'$ be the vectors such that $\cA(g_i)=1$ 
for all $i\in[m'].$ 
\end{itemize}
\item 
If $m'<m/100B^2n,$ proceed to the next $\sigma.$
Else, compute $v_\sigma\in\R^n$ as the maximizer
of the objective function 
$z(v)=\frac1{m'}\sum_{i=1}^{m'}\langle v,g_i'\rangle^2.$
\end{enumerate}
\item Let $v^*$ denote the first vector $v_\sigma$ that achieved 
objective function $(\sigma^2+1/4) + \Delta$ where $\Delta=1/7Br.$
\begin{itemize}
\item If no such $v_\sigma$ was found, let
$V_{t+1}=V_t$ and proceed to the next round.
\item Else let $v_t = v^* -\frac{\sum_{v\in V_t}v\langle v,v^*\rangle}
{\|\sum_{v\in V_t}v\langle v,v^*\rangle \|}$
and put $V_{t+1} = V_t\cup\Set{v_t}.$
\end{itemize}
\end{enumerate}
\end{boxedminipage}
\caption{Reconstruction Attack on Linear Sketches. The algorithm
iteratively builds a subspace $V_t$ that is approximately contained in the
unknown subspace $A.$ In each round the algorithm queries $\cA$ on a sequence
of queries chosen from the orthogonal complement of $V_t.$ As the dimension of
$V_t$ grows larger, the oracle must make a mistake.}
\figurelabel{attack}
\end{figure}

\subsection{Proof of \theoremref{attack}}

\begin{proof}
We may assume without loss of generality that $n=r+90\log(Br)$ by working with
the first $r+90\log(Br)$ coordinates of~$\R^n.$ This ensures that a polynomial
dependence on $n$ in our algorithm is also a polynomial dependence on $r.$

For each $1\le t\le t,$ let $W_t\subseteq A$ be the closest $(t-1)$-dimensional 
subspace to $V_t$ that is contained in $A.$ Formally, $W_t$ satisfies
\begin{equation}
d(V_t,W_t) = \min\set{ d(V_t,W) \colon \dim(W)=t-1, W\subseteq A}\mper
\end{equation}
Note that here we identify
$V_t$ with the subspace that is spanned by the vectors contained in $V_t.$ 
We will maintain (with high probability) that the following invariant is true 
during the attack:
\begin{description}
\item[Invariant at step $t$:] 
\begin{equation}\equationlabel{invariant}
\dim(V_t) = t-1\qquad\text{and}\qquad
d(V_t,W_t)\le 
\frac {t}{20(Bn)^{3.5}\log(Bn)^{2.5}}.
\end{equation}
\end{description}
Note that the invariant holds vacuously at step~$1,$ since $V_1=\{0\}\subseteq
A.$ Informally speaking, our goal is to show that either the algorithm
terminates with a failure certificate or the invariant continues to hold.
Note that whenever the invariant holds in a step~$t,$ we must have
\[
d(V_t,W_t)\le \frac1{20B^{3.5}n^{2.5}\log(Bn)^{2.5}}.
\]
Hence, \lemmaref{stat-dist} shows that for every $\sigma^2\in(0,B],$
\begin{equation}\equationlabel{VtWt}
\|G(V_t^\bot,\sigma^2)-G(W_t^\bot,\sigma^2)\|_\tv \le 
20\sqrt{B n\log(Bn)}\cdot d(V_t,W_t)+\frac1{(Bn)^5}
\le \frac{1}{B^3n^2\log(Bn)^2}\mper
\end{equation}
This observation leads to the following useful lemma.
\begin{lemma}\lemmalabel{correct}
Assume that the invariant holds at step $t.$ Then, if $f$ is
$(\alpha,B)$-correct on $V_t^\bot,$ then $f$ is $(\alpha+\epsilon,B)$-correct on
$W_t^\bot.$
\end{lemma}
\begin{proof}
\equationref{VtWt} implies that for every $\sigma^2\in(0,B],$ 
the statistical distance between $G(V_t^\bot,\sigma^2)$ and
$G(W_t^\bot,\sigma^2)$ is at most $\epsilon.$ Hence, the correctness
conditions from \definitionref{correctness} hold up to an $\epsilon$-loss in
the probabilities.
\end{proof}

Let $E$ denote the event that the empirical estimate $s(t,\sigma^2)$ is
accurate at all steps of the algorithm. Formally:
\[
\forall t\forall\sigma^2\in S:\left| s(t,\sigma^2) - \Pr_{G(V_t^\bot,\sigma^2)}\Set{f(g)=1}\right|
\le\epsilon\mper
\]
\begin{lemma}
$\Pr\Set{E}\ge 1-\exp(-n).$
\end{lemma}
\begin{proof}
The claim follows from a standard application of the Chernoff bound, since we
chose the number of samples $m\gg (Bn/\epsilon)^2\mper$
\end{proof}
\begin{lemma}
\lemmalabel{empirical}
Under the condition that $E$ occurs, the following is true:
If the algorithm terminates in round~$t$ and outputs $G(V_t^\bot,\sigma^2),$
then $G(V_t^\bot,\sigma^2)$ is a failure certificate for~$f.$ Moreover, if the
algorithm does not terminate in round~$t$ and the invariant holds in
round~$t,$ then $f$ is $B$-correct on $W^\bot.$
\end{lemma}
\begin{proof}
The first claim follows directly from the definition of a failure certificate
and the condition that the empirical error given by $s(t,\sigma^2)$ is
$\epsilon$-close to the actual error.

The second claim follows from \lemmaref{correct}. Indeed by the condition $E$
and the assumption that the algorithm did not terminate,
we must have that $f$ is $(2\epsilon,B)$-correct on $V_t^\bot.$ By
\lemmaref{correct}, this implies that $f$ is $(3\epsilon,B)$-correct on
$W_t^\bot.$ Note that $3\epsilon\le 1/10(Bn)^2$ and hence $f$ is correct on
$W_t^\bot.$
\end{proof}

The next lemma is crucial as it shows that the invariant continues to hold
with high probability assuming that $f$ continues to be correct. 
\begin{lemma}[Progress]\lemmalabel{progress}
Let $t\le r.$ Assume that the invariant holds in round $t$ and that $f$ is
$B$-correct on $W_t^\bot.$ Then,  
with probability $1-1/n^2$ the invariant holds in round $t+1.$
\end{lemma}

We will carry out the proof of \lemmaref{progress} in \sectionref{progress}.
But before we do so we will conclude the proof of \theoremref{attack} assuming
that the previous lemma holds.  In order to do so, we argue that if we reach
the final round and the invariant holds, we have effectively reconstructed all
of $A.$ Hence, it must be the case that $f$ is no longer correct on $W_t^\bot$
as shown next.
\begin{lemma}\lemmalabel{final}
Suppose that the invariant holds for $t=r+1.$ Then, $f$ is not $B$-correct on $W_t.$ 
\end{lemma}
\begin{proof}
Since $t=r+1$ and the invariant holds, we have $\dim(V_t)=\dim(W_t)=r.$ On the
other hand $W_t\subseteq A$ and $\dim(A)=r.$ Hence, $W_t=A.$ Therefore,
the function $f$ cannot distinguish between samples from $G(W_t^\bot,2)$ and
samples from $G(W_t^\bot,B).$ Thus, $f$ must make a mistake with constant
probability on one of the distributions. 
\end{proof}
Condition on the event that $E$ occurs. Since $E$ has probability
$1-\exp(-n),$ this affects the success probability of our algorithm only by a
negligible amount. Under this condition, if the algorithm terminates in a
round~$t$ with $t\le r,$ then by \lemmaref{empirical}, the algorithm actually
outputs a failure certificate for $f.$ On the other hand, suppose that we do
not terminate in any of the rounds $t\le r.$ By the second part of
\lemmaref{empirical}, this means that in each round $t$ it must be the
case that $f$ is correct on $W_t^\bot$ assuming that the invariant holds at
step~$t.$ In this case we can apply
\lemmaref{progress} to argue that the invariant continues to hold in round
$t+1.$ Since the invariant holds in step~$1,$ it follows that if the algorithm
does not terminate prematurely, then with probability $(1-1/n^2)^r\ge 1-1/n$ the
invariant still holds at step $r+1.$ But in this case, $W_{r+1}$ is not
correct for $f$ by \lemmaref{final} and hence by \lemmaref{empirical} we output
a failure certificate with probability $1-\exp(-n).$ Combining the two
possible cases, it follows that the algorithm successfully finds a failure
certificate for $f$ with probability $1-2/n.$ This is what is required by
\theoremref{attack}.

It therefore only remains to argue about query complexity and running time. The query
complexity is polynomially bounded in $n$ and hence also in $r$ since we may
assume that $n\le O(r)$ as previously argued.
Computationally, the only non-trivial step is finding the vector $v_\sigma$
that maximizes $z(v)=\frac1{m'}\sum_{i=1}^{m'}\langle v_\sigma,g_i\rangle^2.$
We claim that this vector can be found efficiently using singular vector
computation. Indeed, let $G$ be the $m'\times n$ matrix that has
$g_1,\dots,g_m'$ as its rows. The top singular vector $v$ of $G,$ by definition,
maximizes $\|Gv\|^2= \sum_{i=1}^{m'}\langle g_i,v\rangle^2.$ Hence, it must
also maximize the $z(v).$ This shows that the attack can be implemented in
time polynomial in~$r.$
This concludes the proof of \theoremref{attack}
\end{proof}

\subsection{Proof of the Progress Lemma (\lemmaref{progress})}
\sectionlabel{progress}
Let $t\le r.$ Assume that the invariant holds in round $t.$ Further assume
that $f$ is $B$-correct on $W_t^\bot.$ 
Recall that under these assumptions, by \equationref{VtWt}, for 
every $\sigma^2\in(0,B],$
\begin{equation}\equationlabel{VtWt}
\delta\defeq \|G(V_t^\bot,\sigma^2)-G(W_t^\bot,\sigma^2)\|_\tv 
\le \frac{1}{B^3n^2\log(Bn)^2}\mper
\end{equation}
Our goal is to show that with probability $1-1/n^2,$ the invariant holds in round $t+1.$ 
To prove this claim we will invoke the conditional expectation lemma
(\lemmaref{conditional-main}) in the analysis of our algorithm. 

\begin{lemma}\lemmalabel{sigmatilde}
Assume that $f$ is correct on $W_t^\bot.$
There exists a $\tilde\sigma^2\in S,$ $\Delta\ge
\frac 1{7Br}$ and $u\in V_t^\bot\cap A$
such that for $g\sim G(V_t^\bot,\tilde\sigma^2),$ we have
$\Pr\Set{f(g)=1}\ge 1/60B^2r$ and 
\[
\E \left[\langle u,g\rangle^2 \, \Big|\, f(g)=1\right]
\ge \E\left[\langle u,g\rangle^2\right] + \Delta\mper
\]
\end{lemma}
\begin{proof}
By our assumption, $W_t^\bot$ satisfies the assumptions of the
conditional expectation lemma (\lemmaref{conditional-main}) so that
there exists $u\in U=W_t^\bot\cap A$ and $\sigma\in[3/4,B],$ such
that 
\[
\E_{G(W_t^\bot,\sigma^2)} \left[\langle u,g\rangle^2 \, \Big|\, f(g)=1\right]
\ge \E\left[\langle u,g\rangle^2\right] + \frac 1{4Br}
\]
and $\Pr\Set{f(g)=1}\ge \frac1{40B^2r}\mper$ 
On the other hand, by \equationref{VtWt}, we know that $G(W_t^\bot,\sigma^2)$
is $\delta$-statistically close to $G(V_t^\bot,\sigma^2)$ 
and that $\delta=o(1/B^2n).$
We claim that therefore
\begin{equation}\equationlabel{condex}
\E_{G(V_t^\bot\cap A,\sigma^2)} \left[\langle u,g\rangle^2 \, \Big|\, f(g)=1\right]
\ge \E\left[\langle u,g\rangle^2\right] + \frac 1{6Br}
\end{equation}
and $\Pr\Set{f(g)=1}\ge \frac1{50B^2r}\mper$ The latter
statement is immediate because $f(g)\in\bits$ and hence $\Pr\Set{f(g)=1}$ can
differ by at most $\delta = o(1/B^3r\log(rB))$ between the two distributions. 
This further implies that the statistical distance between the two
distributions under the condition that $f(g)=1$ can only increase by a factor
of $50B^2r.$
Formally, for every distinguishing function $R\colon\R^n\to[0,D],$ we have
\begin{equation}
\equationlabel{expectation-diff}
\left|\E_{G(W_t^\bot,\sigma^2)} \left[R(g) \, \Big|\, f(g)=1\right]
-\E_{G(V_t^\bot\cap A,\sigma^2)} \left[R(g) \, \Big|\, f(g)=1\right]
\right|\le 50B^2r D\delta \mper
\end{equation}
On the other hand,
$\Pr\Set{ \langle u,g\rangle^2>10\ell C}\le \exp\left(-\ell\right).$ 
Hence, we can truncate $\langle u,g\rangle^2$ at $D=10B\log(rB)$ without
affecting either expectation by more than $o(1/Br).$ Hence, by
\equationref{expectation-diff} and the fact that 
$O(B^3r\log(rB))\delta=o(1/Br),$ we have
\begin{equation}
\equationlabel{expectation-diff2}
\left|\E_{G(W_t^\bot,\sigma^2)} \left[\langle u,g\rangle^2 \, \Big|\, f(g)=1\right]
-\E_{G(V_t^\bot\cap A,\sigma^2)} \left[\langle u,g\rangle^2 \, \Big|\, f(g)=1\right]
\right|\le  o\left(\frac 1{Br}\right)\mper
\end{equation}
A similar argument shows that if we change $\sigma^2$ by only $o(1/B^2n^2)$
additively, \equationref{expectation-diff2} continues to hold up to a
insignificant loss in the parameters.  Hence, there exists $\tilde\sigma^2$ in
our discretization for which this claim is true.  Finally, since $u\in W_t\cap
A$ and the invariant holds for $(V_t,W_t),$ we have that $\|P_{V_t}u\|\ge
1-1/B^2n^2.$ Hence, the conclusion of the lemma also holds for some $u\in V_t\cap
A$ up to an additive $o(1/Bn)$ loss in the expectation.
\end{proof}

Informally speaking, the previous lemma suggests that for $\tilde\sigma\in S,$
the vector $v_{\tilde\sigma}$ has exceptionally high objective value.
Moreover, we need to show that \emph{any} vector that has high objective value
must be very close to subspace $V_t^\bot\cap A.$ This is formally argued next.
The result follows from analyzing the top singular vector of the biased
Gaussian matrix obtained by arranging all positively labeled examples into a
matrix. This analysis uses standard techniques which we defer to
\sectionref{net} but rely on in the next lemma.
\begin{lemma}
With probability $1-\exp(-n),$ the vector $v^*$ found by the algorithm in step
$t$ satisfies
\begin{equation}\equationlabel{PUvt}
\|P_{V_t^\bot\cap A}v^*\|^2\ge 1-\frac{1}{20(Bn)^{3.5}\log^4(Bn)}\mper
\end{equation}
\end{lemma}
\begin{proof}
Let $z^*=\max_{\sigma^2\in S}z(v_\sigma)$ denote the maximum objective value
achieved by any $\sigma$ in round $t$ of the algorithm. We will first lower
bound $z^*$ using the information we have about $\tilde\sigma^*$ from the
previous lemma. To this end, we would next like to apply \lemmaref{net} to the 
conditional distribution of $g\sim G(V_t^\bot,\tilde\sigma^2)$ conditioned on the
event that $f(g)=1.$
Note that $g_1',\dots,g'_{m'}$ are uniformly sampled from this distribution.
Furthermore, the probability that $\cA$ outputs $1$ on each sample is at least
$p\ge\Omega(1/B^2n).$ This can be used to show that by a Chernoff bound,
with probability $1-\exp(-n),$ we have that 
\[
m'\ge \frac{pm}{10} \ge \Omega\left(B^{11}n^{10}\log^{15}(n)\right)
\]
We will apply \lemmaref{net} with the following setting of $\gamma$:
\[
\gamma = \frac{1}{(Bn)^{3.5}\log^4(Bn)}\mper
\]
We need to verify
the following conditions of \lemmaref{net}. In doing so let $V=V_t^\bot\cap A$ and
$W=V^\bot = V_t + A^\bot.$ Further, let $\tau=\tilde\sigma^2+1/4$ and
$\Delta$ be the parameter from \lemmaref{sigmatilde}. 
\begin{enumerate}
\item Any unit vector $w\in W$ can be written as $\alpha v + \beta w',$ where
$\alpha^2+\beta^2=1,$ and $v,w'$ are unit vectors with $v\in V_t$ and $w'\in
A^\bot$ is orthogonal to $v.$ But $\E\langle v,g\rangle^2\le 1/4.$ Moreover, the
condition $f(g)=1$ does not bias the distribution along directions inside
$A^\bot.$ Hence, $\E\langle w',g\rangle^2\le \tau.$ It follows that $\E\langle
w,g\rangle^2\le\tau.$ Here we used that $\E\langle w',g\rangle\langle
v,g\rangle=0$ since $g$ can be written as the sum of two independent spherical
gaussians and $v$ and $w'$ are orthogonal.
\item For every $v\in V,$ $w\in W,$ we have 
$\E\langle v,g\rangle\langle w,g\rangle=0.$ This is again because $v$ and $w$ are
orthogonal and $g$ can be written as the sum of two independent spherical
Gaussians.
\item By \lemmaref{sigmatilde}, there exists $v\in V$ such that $\E\langle
v,g\rangle^2\ge \tau + \Delta$ where $\Delta\ge 1/7Br.$
\item Finally, for every $u\in\R^n,$ we claim that
$\xi^2=\Var\langle u,g\rangle^2\le O(B^2\log^2 n)$ as it corresponds to
the fourth moment of a Gaussian with variance $C$ conditioned on an event
of probability $\Omega(1/\poly(n)).$ Any such event can increase the fourth
moment of $O(B^2)$ by at most an $O(\log^2 n)$ factor.
\end{enumerate}
Finally, to apply \lemmaref{net} for the given value of $\gamma,\Delta,$ and
$\xi^2,$ we need the number of samples to be
\[
\Theta\left(\frac{n\log^2(n)\xi^2}{\gamma^2\Delta^2}\right)
\le \Theta\left(B^{11}n^{10}\log^{14}(n)\right)
= o(m')\mper
\]
We have thus verified all conditions of \lemmaref{net}.
It follows that with probability
$1-\exp(-n\log n),$
\[
z^* \ge z(\tilde\sigma) \ge 1 + \frac{\Delta}{14}\mper
\]
On the other hand, let us call a $\sigma^2\in S,$ \emph{bad} if for every unit vector $u\in
V_t^\bot\cap A$ we have for $g\sim G(V_t^\bot,\sigma^2),$
\[
\E \left[\langle u,g\rangle^2 \, \Big|\, f(g)=1\right]
\le \E\left[\langle u,g\rangle^2\right] + \frac{\Delta}{20}\mcom
\]
where $\Delta$ is the parameter from the previous lemma.
We claim that every bad $\sigma^2$ will achieve strictly
smaller objective value, i.e., $z(\sigma)\le 1+\Delta/18,$ with probability
$1-\exp(-n).$ This follows from
\theoremref{bernstein} similarly to how it was used in \lemmaref{net} except
that we now use an upper bound on $\E\langle u,g\rangle^2$ also for $u\in
V_t^\bot\cap A.$ 
Hence, the maximizer of the objective function must correspond to a $\sigma^*$
that satisfies the assumptions of \lemmaref{net} as we previously verified
them for $\tilde\sigma$ up to a constant factor loss in $\Delta.$ Hence, by
\lemmaref{net}, we must have that
$v^*$ satisfies with probability $1-\exp(-n),$
\[
\|P_{V_t^\bot\cap A}v^*\|^2\ge 1-\gamma\mper
\]
\end{proof}

\begin{lemma}\lemmalabel{progress2}
Suppose $(V_t,W_t)$ satisfies the Invariant and suppose the
vector~$v^*$ found in round $t$ satisfies \equationref{PUvt}. Then, the
Invariant holds for $(V_{t+1},W_{t+1}).$
\end{lemma}
\begin{proof}
Let $\gamma=1/(Bn)^{3.5}\log^4(Bn).$
By \equationref{PUvt}, we have that $\|P_Uv^*\|^2\ge 1-\gamma,$ where
$U=V_t^\bot\cap A.$ This means in particular that $\|P_Av^*\|^2\ge 1-\gamma$
and $\|P_{V_t}v^*\|^2\le \gamma.$ This implies that
\[
\|P_A v_t\| \ge 1- O(\gamma)\mper
\]
We claim that this implies that 
\[
d(V_{t+1},W_{t+1})\le d(V_t,W_t)+O(\gamma) \le 
\frac t{(Bn)^{3.5}\log^3(Bn)}+
O\left(\frac1{(Bn)^{3.5}\log^4(Bn)}\right) \le \frac
{t+1}{(Bn)^{3.5}\log(Bn)^3}\mcom
\]
for sufficiently large $n.$ Since this is what we want to show,
it only remains to show that the first
inequality holds. To that end, note that we can write $P_{V_{t+1}}=
P_{V_t}+v_tv_t^T.$ Moreover, we have $P_{W_{t+1}}=P_{W_t} + ww^T$ where $w$ is
some unit vector orthogonal to $W_t$ such that
$\|v_t-w\|\le O(\gamma).$ Hence,
\[
d(V_{t+1},W_{t+1})
= \|P_{V_{t+1}}- P_{W_{t+1}}\|_2
\le \|P_{V_{t+1}}- P_{W_{t+1}}\|_2
+ \|v_tv_t^T - ww^T\|_2
\le d(V_t,W_t)+O(\gamma) \mper
\]
\end{proof}

We have shown that with probability $1-O(\exp(-n)),$ the conditions of
\lemmaref{progress2} are met and in this case the Invariant holds for
$(V_{t+1},W_{t+1}).$ This is what we needed to show in order to conclude the proof of
\lemmaref{progress}.

\subsection{Finding high error certificates via direct products}
\sectionlabel{high-error}

In this section we derive a useful extension of our theorem. Specifically we
show how to find \emph{strong failure certificates}. These will be
distributions as before of the form $G(V^\bot,\sigma^2),$ but now the
algorithm fails with \emph{constant} probability over this distribution rather
than inverse-polynomial.  In fact, our later application to compressed sensing
will need precisely this extension of our theorem.  The proof of this
extension follows from a simple direct product construction which shows how to
reduce the error probability of the oracle at intermediate steps exponentially
while increasing the sample complexity only polynomially. 

\begin{definition}[Strong Failure Certificate]
Let $B\ge 8$ and let $f\colon\R^n\to\bits.$ We say that a pair $(V,\sigma^2)$
is a \emph{$d$-dimensional failure certificate for $f$} if $V\subseteq\R^n$ is $d$-dimensional
subspace and $\sigma^2\in[0,2B]$ such that for some constant $C>0,$ 
we have $n\ge d+10\log(n)$ and moreover:
\begin{itemize}
\item Either $\sigma^2\in[B/2,50B]$ and 
$\Pr_{g\sim G(V^\bot,\sigma^2)}\Set{f(g)=1}\le 2/3,$
\item or $\sigma^2\le 2$ and
$\Pr_{g\sim G(V^\bot,\sigma^2)}\Set{f(g)=1}\ge 1/3.$
\end{itemize}
\end{definition}

The next theorem states that we can efficiently find strong failure
certificates.

\begin{theorem}\label{thm:constant}
Let $B\ge 8.$ Let $A\subseteq\R^n$ be an $r$-dimensional subspace of $\R^n$
such that $n\ge r+90\log(Br).$ Assume that $B\le\poly(n).$ 
Let $f\colon \R^n\to\bits$ satisfying
$f(x)=f(P_A x)$ for all $x\in\R^n.$ Then, there is an algorithm that with
probability $1/3$ and $\poly(B,r)$ adaptive oracle queries to $f$ finds 
a strong failure certificate for $f.$
Moreover, all queries that the algorithm makes are sampled from
$G(V^\bot,\sigma^2)$ for some $V\subseteq\R^n$ and $\sigma^2\in(0,B].$ 
\end{theorem}

\begin{proof}
Let $q=O(\log(Bn)).$ 
Given the function $f\colon\R^n\to\bits$ that is invariant under the subspace
$A\subseteq\R^n,$ consider the function $f^{\otimes q}\colon\R^{qn}\to\bits$
defined as
\[
f^{\otimes q}(x_1,x_2,\dots,x_q) =
\mathrm{Majority}(f(x_1),f(x_2),\dots,f(x_q))\mper
\]
Note that $f^{\otimes q}$ is invariant under the subspace $A^{\otimes q},$
i.e., the $q$-fold direct product of $A$ with itself. Further note that
$\dim(A^{\otimes q})=qr\le qn-90q\log(Br).$ 
Hence, $f^{\otimes q}$ satisfies the assumptions of \theoremref{attack}.

Now, let $G(W^\bot,\sigma^2)$ be the failure certificate returned by the
algorithm guaranteed by \theoremref{attack}. Recall from
\definitionref{complement-gaussian} that a sample $g\sim G(W^\bot,\sigma^2)$
satisfies $g=P_{W^\bot}g_1 + g_2$ where $g_1\sim N(0,\sigma^2)^{qn}$ and
$g_2\sim N(0,1/qn)^{qn}.$
Let $U_1,\dots,U_q$ be coordinate subspaces so that $U_i$ corresponds to the
$i$-th block of $n$ coordinates in $\R^{qn}.$ 
Clearly, we have the decomposition:
\[
P_{W^\bot} = \sum_{i=1}^q P_{U_i} P_{W^\bot}
=\sum_{i=1}^q P_{U_i\cap W^\bot} \mper
\]
Moreover, the subspaces $U_i\cap W^\bot$ are orthogonal for $i\ne j.$ Since
$g_1$ is a spherical Gaussian, this means that the random variables
$P_{U_i\cap W^\bot}g_1$ for $i\in[q]$ are statistically independent of each
other. The same is true for $P_{U_i}g_2.$ Consider therefore the distribution 
$G_i = P_{U_i\cap W^\bot} g_1 + P_{U_i} g_2$ restricted to its $n$ nonzero
coordinates, i.e., we think of $G_i$ as a random variable in $\R^n.$ In
particular, $G(W^\bot,\sigma^2)=(G_1,\dots,G_q)$ and we have established that
this is a product distribution. It remains to argue that each $G_i$ has the
form $G(V^\bot,\sigma^2;1/qn)$ for some subspace $V.$ This is immediate if we
take $V=U_i\cap W^\bot$ thought of as a subspace of $\R^n.$  
The following fact is now an immediate consequence of the Majority function.
\begin{claim}
If $f$ is $9/10$-correct on $G_i$ for a set of at least $2q/3$ indices
$i\in[q],$ then $f^{\otimes q}$ is $1-1/(Bn)^5$ correct on 
$G(W^\bot,\sigma^2).$ 
\end{claim}
\begin{proof}
For $x_1,\dots,x_q$ drawn from $(G_1,\dots,G_q),$ the expected number of
correct answers by $f$ on these samples is at least $2q/3 \cdot 9/10= 3/5.$ The
probability that the number is below $q/2$ is therefore at most $\exp(-\Omega(q))\le
(Bn)^{-5}$ for $q=O(\log(Bn))$ using a standard Chernoff bound. 
\end{proof}
Our claim implies that $f$ cannot be $9/10$-correct on a $1/3$ fraction of the
distributions $G_i,$ for $i\in[q].$ Outputting a random $G_i$ hence completes
the proof of theorem.
\end{proof}

\section{Top singular vector of biased Gaussian matrices}
\sectionlabel{net}

To analyze our algorithm it was necessary to understand the top singular
vector of certain biased Gaussian matrices. In this section we will prove the
necessary lemmas. We start with a standard discretization of the unit sphere.

\begin{lemma}[$\epsilon$-net for the sphere]
\lemmalabel{epsilon-net}
For every $c>0,$ there is a set $N\subseteq\mathbb{S}^{n-1}$ 
of size $|N|\le \exp(O(n\log (1/c)))$ such that
for every unit vector $x\in\R^n,$ there is a unit vector $v\in N$ 
satisfying $\langle x,v\rangle^2\le c.$ 
\end{lemma}

We will need the following simple variant of the Chernoff-Hoeffding bound.

\begin{theorem}[Chernoff-Hoeffding]
\theoremlabel{bernstein}
\theoremlabel{chernoff}
Let the random variables $X_1,\dots,X_m$ be independent 
random variables.
Let $X = \sum_{i=1}^m X_i$ and let $\xi^2=\Var X.$
Then, for any $t>0,$
\[
\Pr\Set{ \left|X-\E X\right| > t} \le
\exp\left(-\frac{t^2}{4\xi^2}\right)\mper
\]
\end{theorem}

The next lemma is an application of the previous bound. We used it earlier 
in the proof of \theoremref{attack}.

\begin{lemma}\lemmalabel{net}
Let $\tau\ge0.$ 
Let $V$ be a subspace of $\R^n.$
Let $G$ be distribution over $\R^n$ such that for $g\sim G$ we have:
\begin{enumerate}
\item For every unit vector $w\in V^\bot,$ we have $\E\langle w,g\rangle^2 \le
\tau.$
\item For every two unit vectors $v\in V,w\in V^\bot,$ we have $\E\langle
v,g\rangle\langle w,g\rangle = 0.$
\item 
The maximum of $\E\langle v,g\rangle^2$ over all unit vectors $v\in V$ is
equal to $\tau+\Delta$ for some $\Delta>1/\poly(n).$
\item For every unit vector $u\in\R^n,$ we have 
$\Var \langle u,g\rangle^2\le \xi^2.$
\end{enumerate}
Let $\gamma>1/\poly(n).$
Draw $m=O\left(\frac{n\log^2(n)\xi^2}{\gamma^2\Delta^2}\right)$ 
i.i.d. samples $g_1,\dots,g_m\sim G$ and let 
\[
u^* = \arg\max_{\|u\|=1}\sum_{i=1}^m \langle g_i,u\rangle^2\mper
\]
Then, with probability $1-\exp(-n\log^2{n}),$ we have $\|P_Vu^*\|^2\ge1-\gamma.$ 
Moreover,
\[
\frac1m\sum_{i=1}^m \langle g_i,u^*\rangle^2\ge \tau + \frac{\Delta}2\mper
\]
\end{lemma}

\begin{proof}
Let $g_1,\dots,g_m\sim G$ be $m$ i.i.d. samples from $G.$
First consider the vector $v\in V$ guaranteed by the second assumption of the
lemma. Let $X = \sum_i \langle v,g_i\rangle^2.$ 
We have $\E X \ge \tau m + \Delta m$ and $\Var X \le n\xi^2.$
Hence, by \theoremref{bernstein},
\[
\Pr\Set{ X\le (\tau+\Delta) m - \frac{\gamma m\Delta}{4}  } \le
\exp\left(-\frac{\Delta^2\gamma^2 m^2}{O(\xi^2m)}\right)
\le \exp\left(-n\log^2 n\right)
\mper
\]
On the other hand, let $u=\alpha v +\beta w$ with $\alpha^2+\beta^2=1$
be any vector such that $v$ is a unit vector in $V$ and $w$ is a unit vector
in $V^\bot.$ Further assume that $\alpha^2\le 1-\gamma.$ Let $Y = \sum_{i=1}^m
\langle u,g_i\rangle^2.$ Note that
\[
\E \langle u,g\rangle^2 
= \alpha^2\E\langle v,g\rangle^2 + \beta^2\E\langle w,g\rangle^2
+\alpha\beta\E\langle v,g\rangle\langle w,g\rangle
= \alpha^2\E\langle v,g\rangle^2 + \beta^2\E\langle w,g\rangle^2\mper
\]
Here we used the assumption that $\E\langle v,g\rangle\langle w,g\rangle=0.$
Hence,
\[
\E \langle u,g\rangle^2 
= \alpha^2\E\langle v,g\rangle^2 + \beta^2\E\langle w,g\rangle^2
\le (1-\gamma)(\tau+\Delta) + \tau/n
= \tau + (1-\gamma)\Delta\mper
\]
In particular, $\E Y \le \tau +(1-\gamma)\Delta.$ Applying \theoremref{bernstein}
again, we have
\[
\Pr\Set{ Y\ge (\tau+\Delta) m - \frac{3\gamma\Delta m}{4}  } \le
\exp\left(-\frac{\Delta^2\gamma^2m^2}{O(\xi^2m)}\right)
\le \exp\left(-n\log^2 n\right)
\mper
\]
Note that there is a margin of $m\Delta/2$
between the bound on $Y$ and the bound on $X.$
Let $M\subseteq\mathbb{S}^{n-1}$ be the set 
$M=\{ u\colon \|P_Vu\|^2\le 1-\gamma\}.$ Further let $\tilde M$ be the set
obtained from $M$ by replacing each member of $M$ with its nearest point in
$N$ where $N$ be the discretization of the unit sphere given by
\lemmaref{epsilon-net} with the setting $c = \gamma\Delta/8.$ Note that $c\ge
1/\poly(n)$ and hence $|N|=\exp(O(n\log n)).$
We claim that
\[
\max_{u\in M}\sum_{i=1}^m \langle g_i,u\rangle^2
\le
\max_{u\in \tilde M}\sum_{i=1}^m \langle g_i,u\rangle^2
+ \frac{\gamma \Delta m}{8}\mper
\]
This is because each squared inner product can differ by at most
$\gamma\Delta/8$ and there are $m$ terms in the summation.
On the other hand, we have by a union bound,
\begin{align*}
\Pr\Set{\max_{u\in \tilde M}\sum_{i=1}^m \langle g_i,u\rangle^2
>  (\tau+\Delta) m - \frac{3\gamma\Delta m}{4}  } 
& \le |N|\cdot
\Pr\Set{ Y\ge (\tau+\Delta) m - \frac{3\gamma\Delta n}{4}  } \\
& \le \exp(O(n\log n))\exp(-\Omega(n\log^2n))\\
& \le \exp(-\Omega(n\log^2(n)))\mper
\end{align*}
We conclude that with probability $1-\exp(-n\log^2n),$
\[
\max_{u\in M}\sum_{i=1}^m \langle g_i,u\rangle^2
\le (\tau +\Delta) m - \frac{3\gamma \Delta m}{4} + \frac{\gamma \Delta m}{8}
\le (\tau +\Delta) m - \frac{5\gamma \Delta m}{8},
\] 
and thus strictly smaller than the global maximum.
This implies that the global maximizer $u^*$ must satisfy
$\|P_Vu^*\|^2\ge1-\gamma.$
\end{proof}

\section{Applications and Extensions}
\sectionlabel{applications}

In this section we derive various applications of our main theorem.
%
\subsection{Randomized Algorithms}
\label{sec:rand}
While we have stated our results in terms of algorithms that output a
(deterministic) function $f$ of the sketch $Ax$, we obtain the same results
for algorithms which use additional randomness to output a randomized function
$f$ of $Ax$. Indeed, the main observation is that our attack never makes the same query twice, 
with probability $1$. It follows that for each possible 
hardwiring of the randomness of $f$ for each possible input, we obtain a deterministic
function, and can apply Theorem \ref{thm:constant}. 

Now consider the attack in \figureref{attack}. In each
round~$t,$ we now allow the algorithm to use a new function
$f_t\colon\R^n\to\bits$ provided that $f_t(x)$ still only depends on $P_Ax.$
Under this assumption, the proof of \theoremref{attack} (and thus also the one
of \theoremref{constant}) carries through the same way as before except that
we replace $f$ by $f_t$ in each round. What is crucial for the attack is only
that the subspace~$A$ has not changed.

We thus have the following theorem. 
%
%
%

\begin{theorem}\theoremlabel{attackRandomized}
Let $B\ge 8.$ Let $A\subseteq\R^n$ be a $r$-dimensional subspace of $\R^n$
such that $n\ge r+90\log(Br).$ Assume that $B\le\poly(n)$. Let $f\colon \R^n\to\bits$ be a randomized algorithm
for which the distribution on outputs $f(x)$ only depends on $P_A x$,
for all $x\in\R^n.$ 
Then, there is an algorithm that given only oracle access
to $f$, which for each possible fixing of the randomness of $f$, 
finds with constant probability a strong failure certificate for $f$. 
The time and query complexity of the algorithm is bounded by $\poly(Br)$. 
Moreover, all queries that the algorithm makes are sampled from
$G(V^\bot,\sigma^2)$ for some $V\subseteq\R^n$ and $\sigma^2\in(0,B]$.  
Moreover, we may assume that for each query $x \sim G(V_t^{\bot}, \sigma^2)$ 
made by the algorithm in round~$t$, the function $f$ is allowed to 
depend on $V_t^{\bot}$. 
\end{theorem}

\subsection{Approximating $\ell_p$ norms}
Our main theorem also applies to sketches that aim to approximate any
$\ell_p$-norm. A randomized sketch~$Z$ for the $\ell_p$-approximation problem
depends only on a subspace $A\subseteq\R^n$ and given $x\in\R^n$ aims to
output a number $Z(x)$ satisfying
\begin{equation}\equationlabel{ellp-approximation}
\|x\|_p\le Z(x)\le C\cdot\|x\|_p\mper
\end{equation}
The next corollary shows that we can find an input on which the sketch must be
incorrect.
\begin{corollary}\label{cor:lp}
Let $1 \leq p \leq \infty$.  
Let $Z$ be a randomized sketching algorithm for approximating the
$\ell_p$-norm which uses a subspace of dimension at most $n-O(\log(Cn)).$
Then, there is an algorithm which, with constant probability, given only
$\poly(Cn)$ oracle queries to $Z(x),$ finds a vector $x\in\R^n$ which violates
\equationref{ellp-approximation}. 
\end{corollary}
\begin{proof}
The result follows from \theoremref{attackRandomized},
applied with approximation factor $B=O(C^2n^2)$ to the function
$f\colon\R^n\to\bits$ which outputs $1$ if $Z(x) \ge C^2n^2$ and $0$
otherwise.  The theorem gives us a strong failure certificate from which we can find
a vector $x$ satisfying either $\|x\|_2\geq8Cn$ and $f(x) = 0$, or,
$\|x\|_2\leq 4\sqrt{n}$ and $f(x) = 1$.  Using the fact that $\|x\|_p/\sqrt{n} \leq
\|x\|_2 \leq \sqrt{n} \|x\|_p,$ it follows that
\equationref{ellp-approximation} is violated.
\end{proof}

\subsection{Sparse recovery with an $\ell_2/\ell_2$-guarantee}
\sectionlabel{recovery}
An $\ell_2/\ell_2$ sparse recovery algorithm for a given parameter $k$, 
is a randomized sketching algorithm which given $Ax$,
outputs a vector $x'$ that satisfies the approximation guarantee 
$\|x'-x\|_2 \leq C \|x_{\mathrm{tail}(k)}\|_2.$ 
Here $x_{\mathrm{tail}(k)}$ denotes $x$ with its top $k$ coordinates (in
magnitude) replaced with $0$. We will show how to find a vector $x$ which
causes the output $x'$ to violate the approximation guarantee, for any $k \geq 1$. 

\begin{theorem}
\theoremlabel{sparse-recovery}
Consider any randomized $\ell_2/\ell_2$ sparse recovery algorithm with
approximation factor $C \le O(\sqrt{n})$, 
sparsity parameter $k \geq 1$, and sketching matrix 
consisting of $r = o(n/C^2)$ rows. 
Then there is an algorithm which, with constant probability, 
given only oracle access to $x'$ finds a vector $x \in \mathbb{R}^n$ which violates
the approximation guarantee with $\poly(n)$ adaptive oracle queries. 
\end{theorem}
\begin{remark}
We note that in general the $r = O(n/C^2)$ restriction cannot be improved, at least for small $k$, 
since there is an upper bound of $O((n/C^2) \cdot k \log(n/k))$. Indeed, with $O(k \log(n/k))$ rows, there is
a deterministic procedure to compute $x'$ with $\|x'-x\|_2^2 \leq 4\|x_{\textrm{tail}(k)}\|_1^2$. By splitting the 
coordinates of $x$ into $n/C^2$ blocks $x^1, \ldots, x^{n/C^2}$ 
and applying this procedure on each block, the total squared error is
$4 \sum_j \|x^j_{\textrm{tail}(k)}\|_1^2 \leq 4 C^2 \|x_{\textrm{tail}(k)}\|_2^2$, which implies
$\|x'-x\|_2 \leq 2 C \|x_{\textrm{tail}(k)}\|_2$. 
\end{remark}
\begin{proof}
It suffices to prove the theorem for $k = 1$, since extending the theorem to larger $k$
can be done by appending $k-1$ additional coordinates to the query vector, each of value $+\infty$. 
We will prove the theorem for each possible fixing of the randomness of function $f$, and so
we can assume the function $f$ is deterministic. 

We use the sparse recovery algorithm $\mathit{Alg}$
to build an algorithm $\mathit{Alg}'$ for the {\sc GapNorm}($B$) problem 
for some value of $B = \Theta(n)$. 
We will use \theoremref{attackRandomized} to argue that 
with constant probability,
$\mathit{Alg}$ must have failed on some query in the attack. 
We use that in each round $\poly(n)$ queries $x$ are drawn from
a subspace Gaussian family $G(V^{\bot}, \sigma^2)$, for certain $V^{\bot}$ and 
$\sigma$ which are chosen by the attacking algorithm, and vary throughout the course of the attack.
Further, we use that the function $f$ can depend on $V^{\bot}$.  

In a given round in the simulation, we have a subspace $V^{\bot}$.
Let $U = V^{\bot} \cap A$ with dimension $r' \leq r$.  
Let $$S = \Set{i \in [n] \mid e_i P_{V^{\bot}}e_i \geq (1-\kappa^2/C^2)^{1/2}},$$
where $\kappa > 0$ is a sufficiently small constant to be determined. 
Notice that tr$(P_{V^{\bot}}) \geq n - r$ and $e_i P_{V^{\bot}}e_i \leq 1$ for all $i$. The following is a simple
application of Markov's bound.
\begin{lemma}\label{lem:markov}
Let $x = 1- \alpha/C^2$ for a constant $\alpha > 0$. 
The number $z$ of indices $i$ for which $e_i P_{V^{\bot}}e_i$ is larger than $x$
is at least $n - C^2 r/\alpha$. 
\end{lemma}
\begin{proof}
$$n- r = \sum_i e_i P_{V^{\bot}}e_i < z \cdot 1 + (n-z) \cdot x = (1-x)z + n x,$$
and so $(1-x)(n-z) < r$, or $z > n - \frac{C^2 r}{\alpha}$. 
\end{proof}
By \lemmaref{markov}, for appropriate $r = \beta n/C^2$, where $\beta > 0$ is a sufficiently small
positive constant depending on $\kappa$, we have that $|S| \geq 2n/3$, which holds
at any round in the attack since tr$(P_{V^{\bot}})$ is always at least $n-r$. 

\paragraph{The attack:}
We design a function $f(Ax)$, which is allowed and
will depend on $V^{\bot}$ (as well as $A$), to solve {\sc Gap-Norm}($B$). Our reduction is deterministic. 
Here
$x \sim G(V^{\bot}, \sigma^2)$
is a query generated during the attack given by \theoremref{attackRandomized}. 

Given $Ax$, the algorithm $\mathit{Alg'}$ first computes the set $S$
with $|S| \geq 2n/3$ with $e_i P_{V^{\bot}}e_i \geq (1-\kappa^2/C^2)^{1/2}$ for all $i \in S$. Notice
that $S$ depends only on $V^{\bot}$. For $i \in S$, let $y^i = x + 4C\sqrt{n} P_{V^{\bot}} e_i$. Given
$V^{\bot}$ and $Ax$ (and $A$), $\mathit{Alg'}$ can compute $Ay^i = Ax + 4C\sqrt{n} A P_{V^{\bot}} e_i$. 

Let $z^i$ be the output of $\mathit{Alg}$ run on input $Ay^i$, for each $i \in S$.  
If $|z^i_i| \geq C \sqrt{n}$ for all $ i \in S$, then $\mathit{Alg'}$ sets $f(x) = 0$, 
otherwise it sets $f(x) = 1$. 

We assume for all queries $x$ and all $\sigma$ chosen during the attack, 
that $\|x\|_2^2 \in [n \sigma^2/2, 2n \sigma^2]$,
which happens with probability $1-1/n^{\omega(1)}$ by standard concentration 
bounds for the $\chi^2$-distribution. 
To analyze this reduction, we distinguish two cases for which we require
correctness for~$f$.  
\paragraph{Case 1:} $\|x\|_2^2 < 4n$. 
By concentration bounds of the $\chi^2$-distribution, we can assume $\sigma^2 \leq 8$. 
Let $i \in S$ and consider $y^i$. Then,
\begin{eqnarray*}
\|y^i_{\textrm{tail}(1)}\|_2 & \leq & \|y^i- 4C\sqrt{n} (e_i^T P_{V^{\bot}} e_i) e_i\|_2 \\
& \leq & \|x\|_2 + 4 C \sqrt{n} (1- (e_iP_{V^{\bot}} e_i)^2)^{1/2} \\
& \leq & 2\sqrt{n} + 4C \sqrt{n} (1-(\sqrt{1-\kappa^2/C^2})^2)^{1/2} \\
& \leq & 2\sqrt{n} + 4\kappa \sqrt{n}\\
& = & (2 + 4 \kappa) \sqrt{n}.
\end{eqnarray*}
By correctness of $\mathit{Alg}$, it follows that for the output $z^i$ of
$\mathit{Alg}$, we have
\begin{eqnarray}\label{eqn:correct}
\|z^i-y^i\|_2 \leq C (2 + 4 \kappa) \sqrt{n}.
\end{eqnarray}
This implies that $z^i_i$ must be at least $C\sqrt{n}$. Indeed, otherwise we would have
\begin{eqnarray*}
\|z^i-y^i\|_2 & \geq & |z^i_i-y^i_i|
 >  4C \sqrt{n} \sqrt{1-\kappa^2/C^2} - |x_i| - C \sqrt{n}
 \geq  2.5 C \sqrt{n},
\end{eqnarray*}
contradicting (\ref{eqn:correct}) for $\kappa$ a sufficiently small constant (and $C$ at least a sufficiently
large constant, which we can assume since it only weakens the correctness requirement of $\mathit{Alg}$). 

Note that we can assume $\mathit{Alg}$ is correct on a query $x \sim G(V^{\bot}, \sigma)$ 
(as otherwise we are already done), 
Hence, $z^i_i$ must be at least $C\sqrt{n}$
simultaneously for every $i \in S$, 
with probability at least $1-1/n^s$ for arbitrarily large $s$. We thus have:
\begin{eqnarray}\label{eqn:c1}
\Pr_{x \sim G(V^{\bot}, \sigma^2)} \Set{f(x) = 0 \mid \|x\|_2^2 < 4n} \geq 1 - \frac{1}{n^s},
\end{eqnarray}

\paragraph{Case 2:} $Bn/4 \leq \|x\|_2^2 \leq 100 Bn$, for a parameter $B = \Theta(n)$. 
Recall that $x$ is obtained by sampling $g_1 \sim N(0, \sigma^2)^n$,
$g_2 \sim N(0,1/4)^n$, and setting $x = P_{V^{\bot}}g_1 + g_2$. Recall that 
$U = V^{\bot} \cap A$, so that we have $P_Ax = P_U g_1 + P_Ag_2$. We can assume that $\sigma^2 \geq \frac{B}{8}$,
as mentioned above, by tail bounds on the $\chi^2$-distribution. 

Associate $U$ with an $r' \times n$ orthonormal matrix ($r' \leq r$) whose rows span $U$. Since
the rows of $U$ are orthonormal, $\|U\|_F^2 \leq r'$, and so by averaging, for at least $2n/3$ 
of its columns $U_j$, we have
$\|U_j\|_2^2 \leq 3r'/n \leq 3r/n$. Let $T = \{j \mid \|U_j\|_2^2 \leq 3r/n\}$. 

Fix a $j \in T$, and consider $y = x + 4C\sqrt{n} P_{V^{\bot}}e_j$. For $x
\sim G(V^{\bot}, \sigma^2)$, 
we start by upper bounding the variation distance between the distributions of random variables 
$P_Ax = P_Ug_1 + P_A g_2$ and 
$P_Ay = P_Ax + 4C \sqrt{n} P_A \cdot P_{V^{\bot}} e_{j} = P_Ug_1 + 4C\sqrt{n} P_U e_{j} + P_Ag_2$. 
The variation distance cannot decrease by fixing $g_2$, so it suffices to upper bound the variation distance
between $P_Ug_1$ and $P_Ug_1 + 4 C \sqrt{n} P_U e_{j}$. Since $P_U$ is a projection matrix, there is a $1$-to-$1$
map from these distributions to $Ug_1$ and $Ug_1 + 4C\sqrt{n} U e_j$, so we upper bound the variation
distance between the latter two distributions. 

By rotational invariance of the Gaussian distribution, 
$U g_1 \sim N(0, \sigma^2 \mathrm{Id}_{r'})$, while 
$U g_1 + 4C \sqrt{n} U e_j \sim N(4 C \sqrt{n} U_j, \sigma^2 \mathrm{Id}_{r'})$, 
where $U_j$ is the $j$-th column of $U$. 
Applying \factref{shifted-gaussians}, and using $\sigma^2 \geq B/8$,
\begin{align*}
\|N(0, \sigma^2 \mathrm{Id}_{r'}) - N(4 C \sqrt{n} U_j, \sigma^2 \mathrm{Id}_{r'})\|_{\textrm{tv}} 
 \leq  \frac{4 C \sqrt{n} \cdot \|U_j\|_2}{\sigma}
& \leq  \frac{\sqrt{8} \cdot 4 C \sqrt{n} \cdot \sqrt{3r/n}} {\sqrt{B}}\\
& 
 =  \frac{8 \sqrt{6} C  \sqrt{\beta n/C^2}}{\sqrt{B}}
 =  \frac{8 \sqrt{6} \sqrt{n\beta}}{\sqrt{B}}.
\end{align*}
By the triangle inequality, for any $j, j' \in T$, 
$$\|N(4 C \sqrt{n} U_j, \sigma^2 \mathrm{Id}_{r'}) - N(4 C \sqrt{n} U_{j'}, \sigma^2 \mathrm{Id}_{r'})\|_{\textrm{tv}} 
\leq \frac{16 \sqrt{6} \sqrt{n\beta}}{\sqrt{B}},$$
and so the variation distance between the distributions of random variables
$P_A(x + 4 C \sqrt{n} P_{V^{\bot}}e_j)$ and 
$P_A(x + 4 C \sqrt{n} P_{V^{\bot}}e_{j'})$ is at most $\frac{16 \sqrt{6} \sqrt{n\beta}}{\sqrt{B}}$. 
For any constant $\gamma > 0$, we can choose the constant $\beta$ sufficiently small so that
for $B = \gamma^2 n$, this variation distance is at most $1/100$. We fix such a $\beta$,
for a $\gamma$ to be determined below. 

Fix an $i \in S \cap T$. 
Consider the output $z^i$ of $\mathit{Alg}$ given $Ay^i$. 
Using that $\|x\|_2 \leq 10\sqrt{Bn} = 10\gamma n$, we have 
\begin{eqnarray*}
\|y^i_{\textrm{tail}(1)}\|_2 & \leq & \|y^i- 4C\sqrt{n} (e_i^T P_{V^{\bot}} e_i) e_i\|_2 \\
& \leq & \|x\|_2 + 4 C \sqrt{n} (1- (e_iP_{V^{\bot}} e_i)^2)^{1/2} \\
& \leq & 10\gamma n + 4C \sqrt{n} (1-(\sqrt{1-\kappa^2/C^2})^2)^{1/2} \\
& \leq & 10 \gamma n + 4\kappa \sqrt{n}.
\end{eqnarray*}
If $\mathit{Alg}$ succeeds,
\begin{eqnarray*}
\|z^i\|_2 & \leq & \|y^i\|_2 + C \|y^i_{\textrm{tail}(1)}\|_2\\ 
& \leq & \|x\|_2 + 4C\sqrt{n} + C(10 \gamma n + 4\kappa \sqrt{n})\\
& \leq & 10 \gamma n + 4C\sqrt{n} + 10C \gamma n + 4\kappa C \sqrt{n}\\
& \leq & \zeta C n,
\end{eqnarray*}
where $\zeta > 0$ is a constant that can be made arbitrarily small by making $\gamma > 0$
arbitrarily small (and assuming $n$ large enough). 
Hence, $\|z^i\|_2^2 \leq \zeta^2 C^2 n^2$. It follows that $|z^i_j| \geq C \sqrt{n}$ for at most
$\zeta^2 n$ values of $j$.
%
%
Now we use the following facts if $\mathit{Alg}$ succeeds. 
\begin{enumerate}
\item $|S| \geq 2n/3$,
\item $|T| \geq 2n/3$,
\item for all $i$, $z^i$ contains at most $\zeta^2 n$ values of $j$ for which $|z^i_j| \geq C \sqrt{n}$
\item for any $j, j' \in T$, the variation distance between the distributions of random variables
$P_A(x + 4 C \sqrt{n} P_{V^{\bot}}e_j)$ and 
$P_A(x + 4 C \sqrt{n} P_{V^{\bot}}e_{j'})$ is at most $1/100$. 
\end{enumerate} 
The first and second conditions imply that $|S \cap T| \geq 2n/3-n/3 = n/3$. 
The third and fourth conditions imply that for any $i \in S \cap T$, if $\mathit{Alg}$ succeeds
(which we can assume), then 
\begin{eqnarray*}
\Pr_{x \sim G(V^{\bot}, \sigma^2)} \Set{|z^i_i| \geq C \sqrt{n}}
 & \leq & \Pr_{j \in S \cap T} \Pr_{x \sim G(V^{\bot}, \sigma^2)}\Set{|z^i_j| \geq
C \sqrt{n}} + \frac{1}{100}\\
 & \leq & \frac{\zeta^2 n}{n/3} + \frac{1}{100}
 <  \frac{1}{10},
\end{eqnarray*}
where the last inequality follows for sufficiently small constant $\zeta > 0$.
Hence, 
\begin{eqnarray}\label{eqn:c2}
\Pr_{x \sim G(V^{\bot}, \sigma )} \Set{f(x) = 0 \mid Bn/4 \leq \|x\|_2^2 \leq
100 Bn} < \frac{1}{10}.
\end{eqnarray}
Combining, (\ref{eqn:c1}) and (\ref{eqn:c2}), for sufficiently large $n$ 
we have for all $V^{\bot}$ and $\sigma$ chosen throughout the
course of the attack,
$$\Pr_{x \sim G(V^{\bot}, \sigma^2)}
\Set{(f(x) = 1 \wedge \|x\|^2 < 4n) \vee (f(x) = 0 \wedge Bn/4 \leq \|x\|_2^2
\leq 100Bn)} < \frac{1}{10}.$$

\paragraph{Wrap-up:}
We have built a function $f$ for {\sc Gap-Norm}($B$) 
which has distributional error less than $1/10$ on $G(V^{\bot}, \sigma^2)$ for
any $V^{\bot}$ and $\sigma$ chosen throughout the course of the attack, whenever $\|x\|_2^2 < 4n$ or
$Bn/4 \leq \|x\|_2^2 \leq 100Bn$. The reduction is deterministic and holds for each setting of the
randomness of $f$. It follows by Theorem \ref{thm:attackRandomized}, 
that with constant probability, with $\poly(n)$ adaptive oracle queries to $f$, we will find a strong 
failure certificate $(V, \sigma^2)$
for $f$ (for each fixing of its randomness). In this case, either $\sigma^2 > B/2$ and $\Pr_{g \sim G(V^{\bot},
\sigma^2)}\Set{f(g) = 1} \leq 2/3$, which
would violate (\ref{eqn:c2}), or $\sigma^2 \leq 2$ and $\Pr_{g \sim
G(V^{\bot}, \sigma^2)}\Set{f(g) \neq 0} \geq 1/3$, which
would violate (\ref{eqn:c1}). It follows that our assumption that the
compressed sensing algorithm $\mathit{Alg}$
succeeded on all queries made was false, which implies that we have found a
query to $\mathit{Alg}$ violating its
approximation guarantee. This completes the proof. 
\end{proof}

\newcommand{\etalchar}[1]{$^{#1}$}

\end{document}